\documentclass[aps,pra,onecolumn,notitlepage,superscriptaddress,showpacs,nofootinbib]{revtex4-1}
\usepackage{amsmath}
\usepackage{amssymb}
\usepackage{amsfonts}
\usepackage[table]{xcolor}
\usepackage{enumerate}
\usepackage{mathrsfs}
\usepackage{graphicx}
\usepackage{epstopdf}
\usepackage{enumerate}
\usepackage{changepage}
\usepackage{bm}
\usepackage{svg}
\usepackage{multirow}
\usepackage{lipsum}
\usepackage{booktabs}
\usepackage{tikz}
\usepackage[utf8]{inputenc}

\usepackage{colortbl}
\usepackage[normalem]{ulem}
\usepackage{tabularx}
\usepackage{makecell}
\usepackage{adjustbox}

\newcolumntype{L}[1]{>{\centering\arraybackslash}p{#1\textwidth}} 
\newcolumntype{R}[1]{>{\centering\arraybackslash}p{#1\textwidth}} 

\newtheorem{theorem}{Theorem}
\newtheorem{definition}{Definition}

\newtheorem{lemma}{Lemma}
\newtheorem{proposition}{Proposition}
\newtheorem{conjecture}{Conjecture}
\newtheorem{example}{Example}
\newtheorem{corollary}{Corollary}

\def\bcj{\begin{conjecture}}
	\def\ecj{\end{conjecture}}
\def\bcr{\begin{corollary}}
	\def\ecr{\end{corollary}}
\def\bd{\begin{definition}}
	\def\ed{\end{definition}}
\def\bea{\begin{eqnarray}}
	\def\eea{\end{eqnarray}}
\def\bem{\begin{enumerate}}
	\def\eem{\end{enumerate}}
\def\bex{\begin{example}}
	\def\eex{\end{example}}
\def\bim{\begin{itemize}}
	\def\eim{\end{itemize}}
\def\bl{\begin{lemma}}
	\def\el{\end{lemma}}
\def\bma{\begin{bmatrix}}
	\def\ema{\end{bmatrix}}
\def\bpf{\begin{proof}}
	\def\epf{\end{proof}}
\def\bpp{\begin{proposition}}
	\def\epp{\end{proposition}}
\def\bqu{\begin{question}}
	\def\equ{\end{question}}
\def\br{\begin{remark}}
	\def\er{\end{remark}}
\def\bt{\begin{theorem}}
	\def\et{\end{theorem}}

\def\squareforqed{\hbox{\rlap{$\sqcap$}$\sqcup$}}
\def\qed{\ifmmode\squareforqed\else{\unskip\nobreak\hfil
		\penalty50\hskip1em\null\nobreak\hfil\squareforqed
		\parfillskip=0pt\finalhyphendemerits=0\endgraf}\fi}
\def\endenv{\ifmmode\;\else{\unskip\nobreak\hfil
		\penalty50\hskip1em\null\nobreak\hfil\;
		\parfillskip=0pt\finalhyphendemerits=0\endgraf}\fi}
\newenvironment{proof}{\noindent \textbf{{Proof.~} }}{\qed}
\def\Dbar{\leavevmode\lower.6ex\hbox to 0pt
	{\hskip-.23ex\accent"16\hss}D}
\makeatletter
\def\url@leostyle{%
	\@ifundefined{selectfont}{\def\UrlFont{\sf}}{\def\UrlFont{\small\ttfamily}}}
\makeatother
\def\bcj{\begin{conjecture}}
	\def\ecj{\end{conjecture}}
\def\bcr{\begin{corollary}}
	\def\ecr{\end{corollary}}
\def\bd{\begin{definition}}
	\def\ed{\end{definition}}
\def\bea{\begin{eqnarray}}
	\def\eea{\end{eqnarray}}
\def\bem{\begin{enumerate}}
	\def\eem{\end{enumerate}}
\def\bex{\begin{example}}
	\def\eex{\end{example}}
\def\bim{\begin{itemize}}
	\def\eim{\end{itemize}}
\def\bl{\begin{lemma}}
	\def\el{\end{lemma}}
\def\bpf{\begin{proof}}
	\def\epf{\end{proof}}
\def\bpp{\begin{proposition}}
	\def\epp{\end{proposition}}
\def\bqu{\begin{question}}
	\def\equ{\end{question}}
\def\br{\begin{remark}}
	\def\er{\end{remark}}
\def\bt{\begin{theorem}}
	\def\et{\end{theorem}}

\def\btb{\begin{tabular}}
	\def\etb{\end{tabular}}

	\newcommand{\nc}{\newcommand}
	
	\def\i{\iota}

	\nc{\bbA}{\mathbb{A}} \nc{\bbB}{\mathbb{B}} \nc{\bbC}{\mathbb{C}}
	\nc{\bbD}{\mathbb{D}} \nc{\bbE}{\mathbb{E}} \nc{\bbF}{\mathbb{F}}
	\nc{\bbG}{\mathbb{G}} \nc{\bbH}{\mathbb{H}} \nc{\bbI}{\mathbb{I}}
	\nc{\bbJ}{\mathbb{J}} \nc{\bbK}{\mathbb{K}} \nc{\bbL}{\mathbb{L}}
	\nc{\bbM}{\mathbb{M}} \nc{\bbN}{\mathbb{N}} \nc{\bbO}{\mathbb{O}}
	\nc{\bbP}{\mathbb{P}} \nc{\bbQ}{\mathbb{Q}} \nc{\bbR}{\mathbb{R}}
	\nc{\bbS}{\mathbb{S}} \nc{\bbT}{\mathbb{T}} \nc{\bbU}{\mathbb{U}}
	\nc{\bbV}{\mathbb{V}} \nc{\bbW}{\mathbb{W}} \nc{\bbX}{\mathbb{X}}
	\nc{\bbZ}{\mathbb{Z}}
	
	\nc{\bA}{{\bf A}} \nc{\bB}{{\bf B}} \nc{\bC}{{\bf C}}
	\nc{\bD}{{\bf D}} \nc{\bE}{{\bf E}} \nc{\bF}{{\bf F}}
	\nc{\bG}{{\bf G}} \nc{\bH}{{\bf H}} \nc{\bI}{{\bf I}}
	\nc{\bJ}{{\bf J}} \nc{\bK}{{\bf K}} \nc{\bL}{{\bf L}}
	\nc{\bM}{{\bf M}} \nc{\bN}{{\bf N}} \nc{\bO}{{\bf O}}
	\nc{\bP}{{\bf P}} \nc{\bQ}{{\bf Q}} \nc{\bR}{{\bf R}}
	\nc{\bS}{{\bf S}} \nc{\bT}{{\bf T}} \nc{\bU}{{\bf U}}
	\nc{\bV}{{\bf V}} \nc{\bW}{{\bf W}} \nc{\bX}{{\bf X}}
	\nc{\ba}{{\bf a}} \nc{\be}{{\bf e}} \nc{\bu}{{\bf u}}
	\nc{\brr}{{\bf r}} \nc{\bx}{{\bf x}}
	
	\nc{\cA}{{\cal A}} \nc{\cB}{{\cal B}} \nc{\cC}{{\cal C}}
	\nc{\cD}{{\cal D}} \nc{\cE}{{\cal E}} \nc{\cF}{{\cal F}}
	\nc{\cG}{{\cal G}} \nc{\cH}{{\cal H}} \nc{\cI}{{\cal I}}
	\nc{\cJ}{{\cal J}} \nc{\cK}{{\cal K}} \nc{\cL}{{\cal L}}
	\nc{\cM}{{\cal M}} \nc{\cN}{{\cal N}} \nc{\cO}{{\cal O}}
	\nc{\cP}{{\cal P}} \nc{\cQ}{{\cal Q}} \nc{\cR}{{\cal R}}
	\nc{\cS}{{\cal S}} \nc{\cT}{{\cal T}} \nc{\cU}{{\cal U}}
	\nc{\cV}{{\cal V}} \nc{\cW}{{\cal W}} \nc{\cX}{{\cal X}}
	\nc{\cZ}{{\cal Z}}
	
	\nc{\hA}{{\hat{A}}} \nc{\hB}{{\hat{B}}} \nc{\hC}{{\hat{C}}}
	\nc{\hD}{{\hat{D}}} \nc{\hE}{{\hat{E}}} \nc{\hF}{{\hat{F}}}
	\nc{\hG}{{\hat{G}}} \nc{\hH}{{\hat{H}}} \nc{\hI}{{\hat{I}}}
	\nc{\hJ}{{\hat{J}}} \nc{\hK}{{\hat{K}}} \nc{\hL}{{\hat{L}}}
	\nc{\hM}{{\hat{M}}} \nc{\hN}{{\hat{N}}} \nc{\hO}{{\hat{O}}}
	\nc{\hP}{{\hat{P}}} \nc{\hR}{{\hat{R}}} \nc{\hS}{{\hat{S}}}
	\nc{\hT}{{\hat{T}}} \nc{\hU}{{\hat{U}}} \nc{\hV}{{\hat{V}}}
	\nc{\hW}{{\hat{W}}} \nc{\hX}{{\hat{X}}} \nc{\hZ}{{\hat{Z}}}
	
	\nc{\hn}{{\hat{n}}}
	
	\def\max{\mathop{\rm max}}
	\def\min{\mathop{\rm min}}

	\def\supp{\mathop{\rm supp}}
	\def\tr{\mathop{\rm Tr}}

	\newcommand{\bra}[1]{\langle#1|}
	\newcommand{\ket}[1]{|#1\rangle}
	
	\newcommand{\ketbra}[2]{|#1\rangle\!\langle#2|}

	\newcommand{\fl}[2]{\left\lfloor\frac{#1}{#2}\right\rfloor}

	\usepackage[
	colorlinks,
	linkcolor = blue,
	citecolor = blue,
	urlcolor = blue]{hyperref}
	\def \qed {\hfill \vrule height7pt width 7pt depth 0pt}
	
	\newcounter{lastnote}

	\newcommand{\lgreen}{\cellcolor[HTML]{67FD9A}}
        \newcommand{\lpurple}{\cellcolor[HTML]{CBCEFB}}
\begin{document}

\author{Kaiyi Guo}
\affiliation{QICI Quantum Information and Computation Initiative,  School of Computing and Data Science,
The University of Hong Kong, Pokfulam Road, Hong Kong SAR, China}	 

\author{Fei Shi}
\email[]{shifei@mail.ustc.edu.cn}
\affiliation{QICI Quantum Information and Computation Initiative, School of Computing and Data Science,
The University of Hong Kong, Pokfulam Road, Hong Kong SAR, China}

\author{You Zhou}
\email[]{you\_zhou@fudan.edu.cn}
\affiliation{Key Laboratory for Information Science of Electromagnetic Waves
 (Ministry of Education), 
Fudan University, Shanghai 200433, China}
\affiliation{Hefei National Laboratory, Hefei 230088, China}



\author{Qi Zhao}
\email[]{zhaoqi@cs.hku.hk}
\affiliation{QICI Quantum Information and Computation Initiative, School of Computing and Data Science,
The University of Hong Kong, Pokfulam Road, Hong Kong SAR, China}	
\title{Approximate $k$-uniform states: definition, construction and applications}
\begin{abstract}
    $k$-Uniform states are fundamental to quantum information and computing, with applications in multipartite entanglement and quantum error-correcting codes (QECCs). Prior work has primarily focused on constructing exact $k$-uniform states or proving their nonexistence. However, due to inevitable theoretical approximations and experimental imperfections, generating exact $k$-uniform states is neither feasible nor necessary in practice. In this work, we initiate the study of approximate $k$-uniform states, demonstrating that they are locally indistinguishable from their exact counterparts unless massive measurements are performed. We prove that such states can be constructed with high probability from the Haar-random ensemble and, more efficiently, via shallow random quantum circuits. Furthermore, we establish a connection between approximate $k$-uniform states and approximate QECCs, showing that Haar random constructions yield high-performance codes with linear rates, vanishing proximity, and exponentially small failure probability while random circuits can't construct codes with linear code rate in low depth. Finally, we investigate the relationship between approximate QECCs and approximate quantum information masking. Our work lays the foundation for the practical application of $k$-uniform states. 
\end{abstract}

\maketitle

\section{Introduction}
An $n$-partite pure state is called $k$-uniform if every $k$-particle marginal is the normalized identity, indicating maximal entanglement across all possible partitions of $k$ versus $n-k$ particles \cite{PhysRevA.69.052330}. Specifically, a  $\fl{n}{2}$-uniform state is referred to as absolutely maximally entangled (AME), indicating maximal entanglement across every bipartition \cite{PhysRevA.86.052335}.
$k$-Unfiorm states and AME states play critical roles in quantum information processing, with applications in quantum error-correcting codes (QECCs) \cite{PhysRevA.69.052330}, quantum secret sharing \cite{PhysRevLett.83.648}, quantum information masking (QIM) \cite{PhysRevA.104.032601}, and quantum simulation \cite{zhao2024entanglement}.

Constructing $k$-uniform states or disproving their existence remains a significant open challenge in quantum information \cite{PRXQuantum.3.010101,PhysRevLett.128.080507}. Although numerous constructions of $k$-uniform states have been developed \cite{feng2017multipartite,goyeneche2014genuinely,li2019k,pang2019two,goyeneche2015absolutely,PhysRevA.97.062326,raissi2020constructions,Zang_2021,Heterogeneous_Systems,10143323}, $k$-uniform states in many cases are either nonexistent or their existence remains unknown \cite{796376,higuchi2000entangled,PhysRevA.69.052330,PhysRevLett.118.200502,huber2018bounds,10718358,ning2025linear}.
For instance, if we denote AME states of $n$ qudits with local dimension $d$ as AME(n,d), it is known that AME(n,2) states do not exist for $n=4$ and $n \geq 7$ \cite{AMEtable}. Similarly, $n$-qubit $(\fl{n}{2}-1)$-uniform states 
are known to be impossible for 
$n=10$
 and $n\geq 13$ \cite{796376,10718358}. The existence of AME(8,4), AME(7,6) and AME(3,11) states remain unresolved \cite{AMEtable}. Furthermore, even though $k$-uniform states can be constructed from orthogonal arrays and quantum Latin squares \cite{Zang_2021,PhysRevA.99.042332}, these methods are combinatorial based and are challenging to implement in physical experiments; moreover, from a practical perspective, generating exact $k$-uniform states in experiments is impossible due to unavoidable noise in quantum information processing, including finite-depth quantum circuits, compilation errors, and physical noise. These theoretical and experimental imperfections restrict the applicability of exact $k$-uniform states in quantum information tasks. Motivated by this, we initiate the systematic study of approximate $k$-uniform states, where each $k$-partite marginal is sufficiently close to the maximally mixed state. 
Our work demonstrates that even when exact 
$k$-uniform states do not exist, their approximate counterparts can still serve as powerful and practical resources for quantum information processing. This bridges the gap between theoretical constructions and experimental realizability.
 


 


Another important question is how to generate 
approximate $k$-uniform states in practice. Fortunately, unlike the exact cases, 
approximate $k$-uniform states can be constructed efficiently using random unitaries or random circuits. 
Exactly generating Haar random unitaries is difficult and costly, but certain random circuits are both sufficient and efficient for approximating Haar randomness. These circuits can serve as an $\epsilon'$-approximate unitary $t$-design in an extremely low circuit depth \cite{laracuente2024approximateunitarykdesignsshallow,schuster2024random}. Recent works in this area have provided powerful tools that benefit the study of 
$k$-uniform states and their applications \cite{Hayden_2006,Sen_1996,Bouland_2018,PhysRevLett.106.180504,Zhou_2022,Leone_2021,Clifford3Design,Zhu_2017,zhu2016cliffordgroupfailsgracefully}.

In this work, we establish a fundamental definition of $\epsilon$-approximate $k$-uniform states. 
First, we derive the non-existence bound on $\epsilon$ based on weight enumerators and provide an existence bound based on numerical optimization; for instance, we demonstrate that AME(4,2) do not exist for $\epsilon< \frac{1}{2\sqrt{19}}\approx 0.115$, while $0.2887$-approximate AME(4,2) states are shown to exist. 
Subsequently, leveraging the concentration of probability measure phenomenon \cite{Hayden_2006}, we prove that a Haar random state serves as an $\epsilon$-approximate $k$-uniform state with high probability under certain constraints. Next, we relax the Haar random requirement to approximate $t$-design, and establish the high-performance construction of approximate uniform states via shallow random circuits with an almost negligible failure probability and proximity under the large $n$ limit. Thereafter, we present our formalism for approximate QECC, relating it to approximate uniform states, and prove the strong performance of approximate QECC under Haar random construction. Finally we illustrate its application to approximate QIM.





\section{Main Results}
In this section we propose our main results on approximate $k$-uniform states and QECC. In the first part, we prove some basic properties of approximate $k$-uniform states, and show that it is impossible to distinguish approximate $k$-uniform states from exact ones locally unless massive measurements are performed, even if these measurements are ideal. Moreover, we prove that some approximate AME$(n,d)$ states do not exist due to violation of shadow inequalities, just like the exact cases. For the second part, we use Levy's lemma to show that a Haar random state is an approximate $k$-uniform state with high probability under some constraint. To get vanishing failure probability and vanishing proximity $\epsilon$ at the same time under large $n$ limit, the parameters have to be chosen carefully. In the third part, we replace Haar random ensemble to $\epsilon'$-approximate $t$-design and their random circuit construction, and show that low depth random circuits suffice to generate $\epsilon$-approximate $k$-uniform states under certain constraints. In the fourth and fifth part, we show the definition and Haar random construction of approximate QECC. In the sixth part we illustrate the performance of random circuit construction of approximate QECC. In the final part, we relate approximate QECC with approximate QIM, highlighting the practical usage of approximate $k$-uniform states and approximate QECC in practical quantum information processing tasks.
\subsection{Results on approximate $k$-uniform states}
Consider an $n$-partite quantum system associated with Hilbert Space $(\mathbb{C}^{d})^{\otimes n}$. 
The qudits are indexed as $[n]:=\{1,2,\cdots,n\}$. 
Let $S\subseteq [n]$ be an index subset in the $n$-partite system. We denote $I_S$ as the identity matrix on subsystem $S$, $I_d$ as the identity matrix of dimension $d$, and define $d_S:= d^{|S|}$. 
For an $n$-qudit pure quantum state $\ket\psi$, we denote $\rho:=\ket\psi\bra\psi$ as its density matrix. The reduced state on a subsystem $S$ is denoted as $\rho_S:=\mathrm{Tr}_{S^c}\rho$. We also use the same notation on operators: $O_S:= \tr_{S^c}O$. In this paper, we use $||\cdot||$ to denote Hilbert-Schmidt norm unless otherwise specified.


A pure state $\ket \psi\in (\mathbb{C}^{d})^{\otimes n}$ is called a \emph{$k$-uniform state} if for any subsystem $S\subseteq [n]$ with size $|S|=k$, the reduced state on $S$ is equal to the normalized identity matrix \cite{PhysRevA.69.052330}:
\begin{equation}
(\ket\psi\bra\psi)_S= \frac {1}{d_S}  I_{S},\qquad \forall S\subseteq [n], \, |S|= k.
\end{equation}
For a $k$-uniform state in $(\mathbb{C}^{d})^{\otimes n}$, it is clear that $k\leq \fl{n}{2}$. Specifically, a $k$-uniform state  with $k=\fl{n}{2}$ is known as  an AME state \cite{PhysRevA.86.052335},  also referred to as an AME$(n,d)$ state. For example, the 3-qubit Greenberger–Horne–Zeilinger (GHZ) state
$\frac{1}{\sqrt{2}}(\ket{000}+\ket{111})$, is a $1$-uniform state, and it is also an AME state. However, $k$-uniform states cannot be perfectly generated in experiments, and many 
$k$-uniform states do not exist, which leads us to explore the concept of approximate $k$-uniform states.

\begin{definition}[\(\epsilon\)-approximate $k$-uniform state]
\label{def:k-u} Let \(\ket\psi\) be a pure quantum state in $(\mathbb{C}^d)^{\otimes n}$, then $\ket\psi$ is an \emph{$\epsilon$-approximate $k$-uniform state} if for any subsystem $S\subseteq[n]$ with size $|S|= k$, the reduced state on $S$  has purity $\epsilon$-close to that of normalized identity matrix:
\begin{equation}\label{eq:k-u}
\mathrm{Tr}\rho_{S}^{2} \leq \frac{1}{d_S} + \epsilon^2.
\end{equation}
Similarly, an $\epsilon$-approximate $\lfloor \frac{n}{2} \rfloor$-uniform state is called an \emph{$\epsilon$-approximate AME$(n,d)$ state}. 
\end{definition}

Note that approximate $k$-uniform state has the following two equivalent definitions:
\begin{enumerate}[(i)]
    \item \label{alternate-def1}For a subsystem $S$ and a pure state $\ket{\psi}$, the eigenvalues of the reduced state   $\rho_S$ can be written as $\{\frac {1+\epsilon_i}{d_S}\}_{i=1}^{d_S}$, where $\epsilon_i$ is the relative deviation of the $i$th eigenvalue. If  $\sqrt{\sum_{i=1}^{d_S} \epsilon_i^2}\leq d_S \epsilon$  for any subsystem $S$ with size $|S|=k$, then $\ket{\psi}$ is an $\epsilon$-approximate $k$-uniform state.
    \item \label{alternate-def2}For a subsystem $S$ and a pure state $\ket{\psi}$, the reduced density matrix $\rho_S$ can be expressed as $\rho_S=\frac{I_S+P_S}{d_S}$. If $\rho_S$ is $\epsilon$-close to the identity under Hilbert-Schmidt norm: $|| \rho_S - \frac 1{d_S} I_S||=\frac{||P_S||}{d_S}\leq \epsilon$ for any subsystem $S$ with size $|S|=k$, then $\ket{\psi}$ is an $\epsilon$-approximate $k$-uniform state.  
\end{enumerate}
It is known that a $k$-uniform state must be a $(k-1)$-uniform state. We can also show this property for approximate $k$-uniform states.
\begin{lemma}\label{lemma:k_1_approximate}
An $\epsilon$-approximate $k$-uniform state is also a $d\epsilon$-approximate $(k-1)$-uniform state.
\end{lemma}
\begin{proof}
    To prove this, we only need to show that Eq.~\eqref{eq:k-u} holds with $|S|=k-1$ by tracing out another qudit from some subsystem $S$. Without loss of generality, suppose the qudit index to be traced out is $1$, $1\in S$, and all eigenvalues of $\rho_S$ are $\{(1+\epsilon_{ij})/d_S\}$ with index $i$ for the subsystem $S\setminus\{1\}$ ranging from $1$ to $d_S/d$ and $j$ for the subsystem $\{1\}$ ranging from $1$ to $d$. According to the alternative Definition~\ref{alternate-def1} of approximate $k$-uniform states, $\sqrt{\sum_{i=1}^{d_S / d} \sum_{j=1}^d \epsilon_{ij}^2} \leq d_S\epsilon$. So the $i$th eigenvalue of $(k-1)$-body marginal $\rho_{S\setminus\{1\}}$ is gained by taking summation over $j$:
    \begin{equation}
    \sum_{j=1}^d \frac{(1+\epsilon_{ij})}{d_S} = \frac{ 1+ \bar\epsilon_i }{d_S /d},
    \end{equation}
    where $\bar\epsilon_i=\frac{\sum_{j=1}^d \epsilon_{ij}}{d}$ is the average deviation. We can prove that the $(k-1)$-body proximity can be bounded by $d\epsilon$:
    \begin{align}
    \sqrt{\sum_{i=1}^{d_S/d}\bar\epsilon_i^2}\leq \sqrt{\sum_{i=1}^{d_S/d} \epsilon_{i,\max}^2}\leq \sqrt{\sum_{i=1}^{d_S/d} \sum_{j=1}^{d} \epsilon_{ij}^2}\leq d_S\epsilon=\frac{d_S}{d}(d\epsilon),
    \end{align}
   where  $\epsilon_{i,\max} := \max_{j\in \{1,2,\cdots, d\}} \epsilon_{ij}$. So the inequality holds for any $|S|=k-1$, the state is a $d\epsilon$-approximate $k-1$-uniform state. We can generalize this lemma to show that it is also a $d^l\epsilon$-approximate $k-l$-uniform state for any $k> l\geq 0$.
\end{proof}

\vspace{0.4cm}

From the definition we can see that the reduced density matrix $\rho_S$ of an $\epsilon$-approximate $k$-uniform state is close to the maximally mixed state. To relate the $\epsilon$-approximate $k$-uniform states to the exact ones, we can use the relation between relative entropy and basic statistical inference to prove that approximate uniform states are locally indistinguishable unless massive measurements are performed, even if these measurements are ideal:
\begin{lemma}[Local indistinguishability of approximate $k$-uniform states]\label{lemma:local-indistinguish}
For sufficiently small $\epsilon$, the reduced state of an $\epsilon$-approximate $k$-uniform state on any subsystem $S$ of size $k$ can not be distinguished from maximally mixed state $\frac{1}{d_{S}} I_{S}$ unless $N\gg  \frac{2}{d_{S}\epsilon ^{2}}$ measurements are performed.
\end{lemma}
\begin{proof}
It is known that if our hypothesis to a quantum state is $\sigma$ and the real quantum state is $\rho$, to prove that our hypothesis is wrong, we need $N\gg \frac{1}{S(\rho | | \sigma)}$ measurements\cite{Vedral_1997} 
where $S(\rho | | \sigma)  := -\mathrm{Tr} (\rho (\log \rho - \log \sigma))$ is the quantum analog of relative entropy\cite{Vedral_2002}. Suppose the subsystem where the measurements are to perform is $S$. In this case $\sigma= \frac{1}{d_{S}}I_{S}$ is the hypothesized maximally mixed state, and $\rho= \frac{1}{d_{S}}(I_{S}+P_{S})$ is the reduction of $\epsilon$-approximate $k$-uniform state on subsystem $S$, where $|| P_{S} | | \leq d_{S}\epsilon$ according to equivalent Definition~\ref{alternate-def2} of approximate $k$-uniform states. Using the definition of relative entropy, we have
\begin{equation}
\begin{aligned}
S(\rho | | \sigma)  & = -\mathrm{Tr} (\rho (\log \rho - \log \sigma)) = \log d_{S} - S(\rho),
\end{aligned}
\end{equation}
where $S(\rho):= -\mathrm{Tr} \rho \log \rho$ is the von-Neumann entropy. Expanding $S(\rho)$ into series of fluctuation $P_{S}$,
\begin{equation}
\begin{aligned}
S(\rho) & = \log d_{S} - \frac{1}{d_{S}} \mathrm{Tr} \left[(I_{S} + P_{S}) \left(  P_{S} - \frac{1}{2}P_{S}^{2} + \frac{1}{3}P_{S}^{3} + \dots \right)\right] \\
  & = \log d_{S} - \frac{1}{d_{S}} \mathrm{Tr} \left( P_{S}+\frac{1}{2} P_{S}^{2}     \right) + O( d_{S}^{2}\epsilon^{3} ) \\
  &  \geq \log d_{S} - \frac{1}{2}d_{S}\epsilon ^{2},
\end{aligned}
\end{equation}
where we used the property of $P_{S}$ that $\mathrm{Tr}P_{S}=0$ and $\mathrm{Tr} P_{S}^{2}\leq (d_{S}\epsilon)^{2}$, and the condition that $\epsilon$ is sufficiently small.
From this we conclude that we need at least $N$ measurements to disprove our hypothesis that $\rho$ is maximally mixed, where
\begin{equation}
N \gg \frac{1}{ S(\rho | | \sigma)} = \frac{1}{\log d_{S}- S(\rho)} \geq \frac{2}{d_{S}\epsilon ^{2}}.
\end{equation}
\end{proof}

In Definition~\ref{def:k-u} we proposed the definition of approximate $k$-uniform states. However, one should note that such states may still be non-existent under certain cases even if we relaxed the restriction. Next, we present our findings on the non-existency of approximate AME states within the framework of the shadow inequalities.

 For a state $\rho$, the  shadow inequalities are \cite{817508}:  
\begin{equation}
s_{T}(\rho)=\sum_{S\subseteq [n]} (-1)^{|S \cap T|} \mathrm{Tr}\rho_{S}^{2} \geq 0,\quad \forall \ T\subseteq [n].
\end{equation}
It is known that  the shadow inequalities can be used to show the non-existence of AME states \cite{huber2018bounds}; in fact, most non-existence results of AME states are given by this inequality. We take AME$(4,2)$ as an example. Suppose $\ket{\psi}$ is an AME$(4,2)$ state. Since all its reductions on $k\leq 2$ parties are maximally mixed, we have $\tr\rho_S^2=\frac{1}{2^{|S|}}$ for $|S|\leq 2$. Observe that for a pure state, the property $\tr\rho_S^2=\tr \rho_{S^c}^2$ holds, enabling us to compute all the local purities of the AME states. By replacing $\rho$ with $\ket\psi\bra\psi$ and taking $T=[4]$ in the shadow enumerator $s_{T}(\rho)$, we have
\begin{equation}
s_T(\rho)= {4\choose 0} - {4\choose 1} \frac{1}2 + {4\choose 2} \frac{1}4 - {4\choose 3} \frac 12 + {4\choose 4}= -\frac 12<0,
\end{equation}
which violates the shadow inequality and proves the non-existence of AME$(4,2)$ states.

Note that  the shadow enumerator $s_T(\rho)$ is a continuous function on the quantum state space. If $s_T(\rho)<0$,  there exists an $\epsilon>0$ such that for all  $\tilde\rho$ within the $\epsilon$-ball $B_\epsilon(\rho)$ centered around $\rho$, we have $s_T(\tilde\rho)<0$. The following theorem gives an explicit upper bound on $\epsilon$.

\begin{theorem}
    If AME$(n,d)$ states do not exist due to violating the shadow inequality $s_T(\rho)<0$, $\epsilon$-approximate AME$(n,d)$ states do not exist for $\epsilon<\sqrt{\frac{-s_T(\rho)}{f(d,n,t)}}$, where 
    
    \begin{equation}f(d,n,t)= \sum_{l=0}^{\lfloor \frac{t}2 \rfloor} \sum_{k'=0}^{n-t} {t\choose 2l} {n-t \choose k'}d^{2\max\left(\lfloor \frac n2\rfloor -(2l+k'), (2l+k')- \lceil \frac n2 \rceil\right)}
    \end{equation} 
    is a function dependent on local dimension $d$, the system number $n$, and the size of selected system $t:=|T|$.
\end{theorem}

\begin{proof}
For an exact AME$(n,d)$ state, 
it must have
\begin{equation}
\tr\rho_S^2 = \frac{1}{d^{\min(|S|,n-|S|)}}, \, \forall S\subseteq [n].
\end{equation}
If an $\epsilon$-approximate AME$(n,d)$ state $\tilde\rho$ exist, then due to non-negativity of the shadow enumerator,
\begin{equation}
s_T(\tilde\rho)\ge 0.
\end{equation}
 We can use the Lemma~\ref{lemma:k_1_approximate} to show that $\epsilon$-approximate $\lfloor \frac n2\rfloor$ uniform states are $\epsilon'$-approximate $|S|$-uniform states with $\epsilon'=\epsilon d^{\lfloor \frac n2\rfloor - |S|}$ when $|S|\leq \lfloor \frac n2 \rfloor$. Next, we use this conclusion and $\tr \rho_S^2 =\tr \rho_{S^c}^2$ for pure state density matrix $\rho$ when $|S|\geq \lceil \frac n2 \rceil$, we have for any $S\subseteq [n]$ 
\begin{equation}
\frac{1}{d^{\min(|S|,n-|S|)}}\leq \mathrm{Tr}\tilde{\rho}_{ S }^{2}\leq \frac{1}{d^{\min(|S|,n-|S|)}} + d^{2\max(\lfloor \frac n2\rfloor - |S|, |S|- \lceil \frac n2 \rceil)}\epsilon^{2},
\end{equation}
where the lower bound is trivial and gained when $\tilde{\rho}_S$ is a maximally mixed state, and the upper bound is from Definition~\ref{def:k-u} and Lemma~\ref{lemma:k_1_approximate}. Without loss of generality, assume that $\tr\tilde\rho_S^2= {d^{-\min(|S|,n-|S|)}} +  d^{2\max(\lfloor \frac n2\rfloor - |S|, |S|- \lceil \frac n2 \rceil)}\epsilon_S^2$ with small deviations $0\leq\epsilon_S^2\leq \epsilon^2$ for any $S\subseteq [n]$.

The shadow enumerator of approximate AME states are close to the exact ones. It is possible to give an upper bound of shadow enumerator on approximate AME$(n,d)$ states:
\begin{equation}\label{eq:s_t_tidle}
\begin{aligned}
s_T(\tilde{\rho}) &= \sum_{S\subseteq [n]} (-1)^{|S\cap T|} \tr \tilde{\rho}_S^2\\
&=  \sum_{S_1\subseteq T} (-1)^{|S_1|} \sum_{S_2\subseteq T^c} \tr \tilde\rho^2_{S_1\cup S_2}\\
&= \sum_{k=0}^{t}\sum_{k'=0}^{n-t}(-1)^k{t\choose k}{n-t \choose k'} \left[\left(\frac{1}{d}\right)^{\min(k+k', n-k-k')} + d^{2\max(\lfloor \frac n2\rfloor -(k+k'), (k+k')- \lceil \frac n2 \rceil)}\epsilon_S^2\right]\\
&= s_T(\rho)+\sum_{k=0}^{t} \sum_{k'=0}^{n-t}(-1)^k {t \choose k} {n-t \choose k'}d^{2\max(\lfloor \frac n2\rfloor -(k+k'), (k+k')- \lceil \frac n2 \rceil)} \epsilon_S^2 \\ 
&\leq s_T(\rho) + \sum_{l=0}^{\lfloor \frac {t}2 \rfloor} \sum_{k'=0}^{n-t} {t \choose 2l} {n-t \choose k'} d^{2\max(\lfloor \frac n2\rfloor -(2l+k'), (2l+k')- \lceil \frac n2 \rceil)} \epsilon^2\\
&= s_T(\rho) + f(d,n,t)\epsilon^2.
\end{aligned}
\end{equation}
Note that in the second equality of Eq.~\eqref{eq:s_t_tidle}, we split $S$ into disjoint two parts $S_1:= S\cap T$ and $S_2:= S\cap T^c$, then do summation separately.
In the third equality of Eq.~\eqref{eq:s_t_tidle}, we take $|T|=t$ and use Definition~\ref{def:k-u}.
In the fourth equality of Eq.~\eqref{eq:s_t_tidle}, we identify the first term to be shadow enumerator of the hypothetical AME$(4,2)$ state $\rho$. 
In the inequality of Eq.~\eqref{eq:s_t_tidle}, we gain an upper bound by taking $k=2l$, removing all negative terms with $k=2l+1$ for integer $l$ and take $\epsilon_S=\epsilon$ for the rest. If $s_T(\rho)<0$ and $\epsilon < \sqrt{\frac {-s_T(\rho)}{f(d,n,|T|)}}$, the inequality gives the violation $s_T(\tilde\rho)<0$, so the approximate AME states do not exist.
\end{proof}
\vspace{0.4cm}

Next we use the simplest case as an example. For AME(4,2) states, we already know that exact AME(4,2) states do not exist since $s_{T}(\rho)=-\frac{1}{2}<0$ for $|T|=[4]$. In the case of $\epsilon$-approximate AME(4,2) state, the shadow inequality tells us
\begin{equation}
\epsilon^2 \geq \frac{\frac 12}{1\cdot 2^4 + 6\cdot 2^0 + 1\cdot 2^4}=\frac 1{76},
\end{equation}
so $\epsilon$-approximate AME(4,2) states do not exist for any $\epsilon< \frac{1}{2\sqrt{ 19 }}\approx 0.115$. In Section.~\ref{num} we used numerical optimization and found that there exists approximate AME$(4,2)$ state with $\epsilon\approx 0.2887$.

\subsection{Haar random unitary construction of approximate $k$-uniform states}\label{sec:haar-approximate-uniform}
Next we show our results on generating $\epsilon$-approximate $k$-uniform states from Haar random unitary. Previous works \cite{Page_1993} have fully investigated the eigenvalue distribution of reduced Haar random states. 
To study the property of Haar random states within the framework of $k$-uniform states, we can relax the condition to approximate $k$-uniform states. With Definition~\ref{def:k-u} of approximate $k$-uniform states, we first show that under certain conditions, a Haar random state $\ket\psi$ is an $\epsilon$-approximate $k$-uniform state with high probability. 

\begin{theorem}
\label{thm:approx-uniform-prob}
A Haar random state $\ket\psi\in (\mathbb C^d)^{\otimes n}$ is an $\epsilon$-approximate $k$-uniform state with success probability $\mathrm{Prob}_H\mathrm{(success)}$ greater than
\begin{equation}
\mathrm{Prob}_H\mathrm{(success)}\geq 1- {n\choose k } 2\exp \left(-\frac{d^{n}\delta^{2}}{72\pi^{3}\log 2}\right) ,
\end{equation}
where
\begin{equation}
\begin{aligned}
\delta &:=\epsilon^2-\Delta, \\
\Delta &:= \frac{d^{k}+d^{n-k}}{d^{n}+1}- \frac{1}{d^{k}},
\end{aligned}\label{def:delta-and-Delta}
\end{equation}
given the constraint that $\epsilon^2\geq 2\Delta$.\label{constraint} $\mathrm{Prob}_H$ refers to the probability under Haar measure.
\end{theorem}
Before giving the proof of Theorem~\ref{thm:approx-uniform-prob}, we first do basic asymptotic analysis on the result. We define the parameter $\alpha:= \frac kn$ with constraint $0<\alpha<\frac 12$. Note that $\Delta=\frac{1}{d^{n-k}}+ O( \frac{1}{d^{n}} \max( \frac{1}{d^{k}}, \frac{1}{d^{n-k}}))$ under large $n$ limit. Since $k<n$, we take $\epsilon\geq \sqrt{ 2\Delta }=\sqrt{ 2 } (d^{- (n-k)/2 }+O(d^{- (2n-k)/2}))$, and use the inequality on combinatorial numbers 
\begin{equation}
    \log {n\choose n\alpha}\leq n S(\alpha)+ \log(n+1)-1,\label{eq:log-comb-number}
\end{equation}
where $S(\alpha):=-\alpha\log \alpha -(1-\alpha)\log (1-\alpha)$ is the Shannon entropy. We have
\begin{equation}
\mathrm{Prob}_H(\mathrm{success}) \geq 1-\exp \left[nS(\alpha)+ \log(n+1) + (\log2-1) - \frac{d^{n}}{72\pi^{3}\log 2}\left( \epsilon ^{2} - \frac{1}{d^{n-k}} + O\left(  \frac{1}{d^{2n-k}}\right)\right)^{2}\right].\label{eq:H-succ}
\end{equation}
We prefer the success probability close to 1, which is equivalent to decreasing the value in the exponential in Eq.~\eqref{eq:H-succ} and make it negative in the large $n$ limit. We can achieve this goal by choosing the parameters carefully. The last term is negative and has factor $d^n$, which contributes significantly in the probability concentration in large $n$ limit. If the proximity satisfies the constraint $\epsilon=\omega(d^{- (n-k)/2})$ and $\epsilon=\omega(d^{-n/4})$, we have $\epsilon^2\gg \frac{1}{d^{n-k}}$ and $d^n\epsilon^4\gg 1$. This allows us to remove the other non-leading terms in the exponential, and say that the success probability is asymptotically greater than $1-e^{-cd^n\epsilon^4}$ for some constant $c>0$, which is doubly-exponentially close to 1 in the large $n$ limit. We have the following corollary: 

\begin{corollary}
A Haar random state $\ket\psi\in (\mathbb C^d)^{\otimes n}$ is an $\epsilon$-approximate $(\alpha n)$-uniform state with failure probability satisfying
\begin{equation}
 \gamma := -\log \mathrm{Prob}_H(\mathrm{fail})\gtrsim  \frac{d^{n}}{72\pi^{3}\log2}\epsilon ^{4}
\end{equation}
as $n\to \infty$, given the constraint that $\epsilon= \omega(d^{-n(1-\alpha)/2})$. The failure probability vanishes quickly under the constraint $\epsilon=\omega(d^{-n/4})$. Here $\mathrm{Prob}_H(\mathrm{fail})=1-\mathrm{Prob}_H(\mathrm{success})$ is the probability that Haar random state fails to be an $\epsilon$-approximate $(\alpha n)$-uniform state. We use $\gtrsim$ notation to represent ``\emph{asymptotically greater or equal to}" relation; more specifically, $f(n)\gtrsim g(n)$ represents for $\lim_{n\to\infty} \frac{f(n)}{g(n)}\geq 1$.
\end{corollary}
In practical cases, we often want failure probability to vanish (i.e. $\gamma=\omega(1)$) while $d_S\epsilon$ is also vanishingly small (i.e. $d_S \epsilon=o(1)$), or the state is locally indistinguishable unless exponentially many experiments are performed ($d_S\epsilon^2=o(1)$ according to Lemma~\ref{lemma:local-indistinguish}) in the large $n$ limit. Without loss of generality, suppose that in the large $n$ limit $\epsilon\sim d^{n\lambda}$ and $\gamma\sim d^{n\mu}$ with negative $\lambda$ and positive $\mu$ being constants, note that $d_S=d^{k}=d^{n\alpha}$, these conditions and constraints leads us to the following linear programming problem:
\begin{equation}
\left\{\begin{aligned}
&\mu\geq 1+4\lambda> 0, &(\text{vanishing failure probability}) \\
& \lambda+\alpha <0 ,&(\text{vanishing proximity})\\
&\frac12 > \alpha > 0,&(\text{linear uniformity rate}) \\
&\lambda > -\frac12 + \frac 12\alpha &(\text{constraint})\\
&\alpha + 2\lambda<0, &(\text{local indistinguishable})
\end{aligned}\right.\label{eq:constraint}.
\end{equation}
 We can make a plot on the $(\alpha,\lambda)$ coordinate system to show the behavior of random construction (note that the horizontal axis is related to $k$ and vertical axis related to $\epsilon$). Within this region, $(\alpha n)$-uniform states can be easily generated with linear uniformity, high proximity and vanishing failure probability under large $n$ limit, see FIG.~\ref{fig:high-perf-area}.

\begin{figure}[h]
\centering
\includegraphics[width=0.8\textwidth]{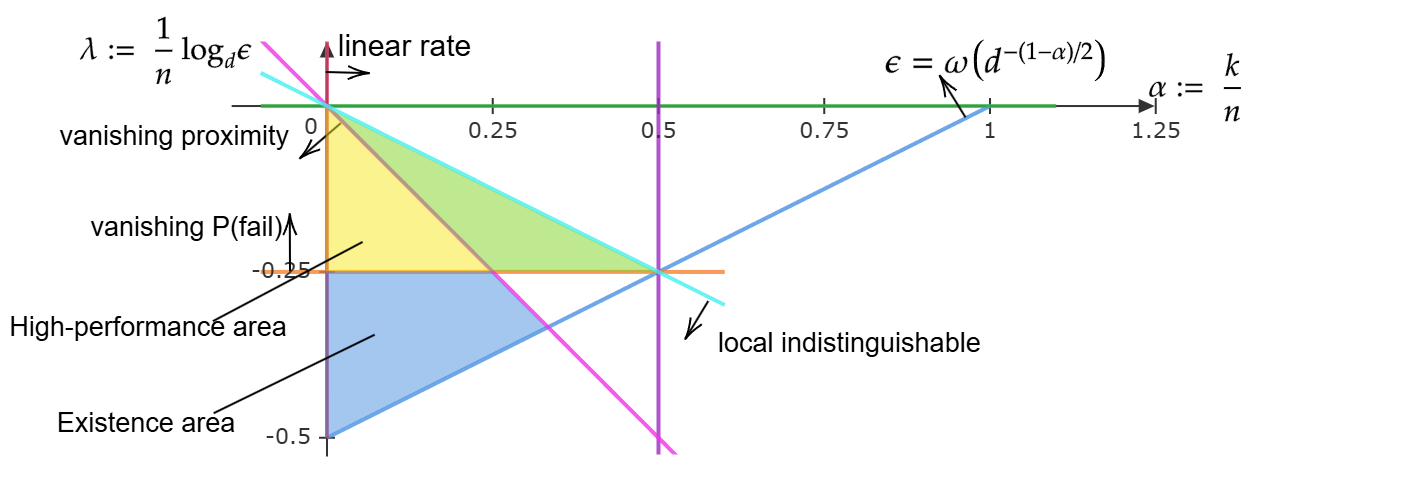}
\caption{
The performance of Haar random construction of $\epsilon$-approximate $k$-uniform states under large $n$ limit. The horizontal axis refers to the uniformity rate $\alpha:=\frac kn$ and the vertical axis refers to the averaged logarithmic error on each qudit $\lambda:=\frac 1n \log_d\epsilon$. The lines represent restrictions from Eq.~\eqref{eq:constraint}. The yellow area represents the high-performance area, where the approximate uniform states can be generated with linear uniformity rate, vanishing proximity and vanishing failure probability; the blue area represents the states that have better properties but may not be Haar-randomly constructed with high probability; the green area represents the states that do not have vanishingly small proximity but is both constructible and locally indistinguishable from exact $(\alpha n)$-uniform states unless exponentially many local measurements are performed.
}
\label{fig:high-perf-area}
\end{figure}

To prove Theorem~\ref{thm:approx-uniform-prob}, we need to provide two facts about purity function first: the average under Haar measure and the Lipschitz constant. Using Levy's Lemma, we can derive the probabilistic result on one subsystem $S$ with $|S|=k$. Then we use de-Morgan's rule to get our final result.
\subsubsection{Proof of Average Purity}
Define the purity function of $\ket\psi$ with density matrix $\rho=\ket\psi\bra\psi$ for some subsystem $A\subseteq [n]$ as $f_A(\psi) := \tr { \rho_A^2}$, and we denote $B:= A^c$ so that $A,B$ is a bipartition of the total system. We can calculate the Haar average of this function:
\begin{lemma}\label{avg-pur}
    If a circuit ensemble $\mu$ is a $2$-design on the composite system $AB$, then the average subsystem purity with respect to $\mu$ is:
\begin{equation}
\mathbb {E}_{U\sim \mu} [f_A(\psi)] = \frac{d_A+d_B}{d_Ad_B+1},
\end{equation}
where $d_A$ and $d_B$ is the dimension of subsystem A and B respectively.
\end{lemma}
Since the Haar measure can be considered as an $\infty$-design, this lemma surely applies to Haar random measure. The lemma can also be found in the related literatures on Haar measure \cite{Mele_2024}.

\begin{proof}
Using the results of Haar-random expectation on $D:= d_Ad_B$ dimensional space with moment 2:
\begin{equation}
\mathbb{E}_{U\sim \mu}[U^{\otimes 2} O U^{\dagger \otimes 2}] = c_{\mathbb{I}}\mathbb{I}+c_{\mathbb{F}}\mathbb{F},
\end{equation}
where the $\mathbb I$ and $\mathbb F$ are identity and exchange operators respectively. The coefficients $c_{\mathbb I},c_{\mathbb F}$ are given by
\begin{align}
c_{\mathbb{I}} & = \frac{ \mathrm{Tr}O- \frac{1}{D}\mathrm{Tr}(\mathbb{F}O) }{D^{2}-1};  \\
c_{\mathbb{F}}  & = \frac{ \mathrm{Tr}(\mathbb{F}O) - \frac{1}{D}\mathrm{Tr}O }{D^{2}-1}.
\end{align} 
Using the property of 2-design, we have
\begin{equation}
    \begin{aligned}
\mathbb{E}_{U\sim \mu}[f_A(\psi)]  & = \mathbb{E}_{U}\mathrm{Tr}\left[(U^{\otimes 2}(\ket{0} \bra{0} _{A} \otimes  \ket{0} \bra{0} _{B})^{\otimes 2} U^{{\dagger}\otimes 2}) (\mathbb{F}_{A}\otimes  \mathbb{I}_{B})\right] \\
  & =\mathrm{Tr} \left[  \frac{ 1- \frac{1}{D} }{D^{2}-1} (\mathbb{I}+\mathbb{F})(\mathbb{F}_{A}\otimes  \mathbb{I} _{B}) \right] \\
  & = \frac{1}{D(D+1)}\mathrm{Tr} [\mathbb{F}_{A}\otimes \mathbb{I}_{B} + \mathbb{I}_{A}\otimes  \mathbb{F}_{B}] \\
  & = \frac{d_{A}d_{B}^{2}+d_{B}^{2}d_{A} }{d_{A}d_{B}(d_{A}d_{B}+1)}  =\frac{d_{A}+d_{B}}{d_{A}d_{B}+1},
\end{aligned}
\end{equation}
where we denoted $\mathbb I_S,\mathbb F_S$ as the identity and exchange operator on the two copies of subsystem $S$, and since the total system is composed from $A,B$, we have $\mathbb I= \mathbb I_A\otimes \mathbb I_B$, $\mathbb F=\mathbb F_A\otimes \mathbb F_B$.
\end{proof}
\subsubsection{Proof of Upper Bound on Lipschitz Constant}
Next we find an upper bound for the Lipschitz constant of purity function:
\begin{lemma}
    The Lipschitz constant of purity function is upper bounded by $4$.\label{lipschitz-pur}
\end{lemma}
\begin{proof}
Using Lemma III.8 in \cite{Hayden_2006} we know that the Lipschitz constant of the function $g(\varphi)=\sqrt{\tr(\varphi_A^2)}$ is upper bounded by $2$. In our case $f_A(\varphi)=g^2(\varphi)$, which gives that for any $\varphi,\psi$ with $f_A(\varphi)\geq f_A(\psi)$,
\begin{equation}
\begin{aligned}
    f_A(\varphi)-f_A(\psi) &= g^2(\varphi)-g^2(\psi)\\ &= (g(\varphi)-g(\psi) )(g(\varphi)+g(\psi))\\
     &\leq 2(g(\varphi)-g(\psi))\\
     &\leq 4|| \ket\varphi-\ket\psi||_2,
\end{aligned}
\end{equation}
where we have used the result that $g(\varphi)\leq 1$. So the Lipschitz constant is upper bounded by $4$.

\end{proof}
\subsubsection{Final Proof}

With the mean and Lipschitz constant of the purity function,
taking $n$ qudits and with each local dimension $d$. Select the subsystem size $|A|=k$, so $d_{A}=d^{k}$ and $d_{B}=d^{n-k}$, we can prove the concentration result (Theorem~\ref{thm:approx-uniform-prob}) on Haar random state $\ket{\psi}$ on $d_{A}d_{B}$-dimensional Hilbert space:

\begin{proof}{(Proof of Theorem~\ref{thm:approx-uniform-prob})}

By applying Levy's lemma (see Lemma III.1.1 in \cite{Hayden_2006}), we know the concentration of purity function under Haar measure: for any Haar random state $\ket\psi\in (\mathbb C^d)^{\otimes n}$, the probability of the event that purity function deviates from its mean value is greater than $\delta$ is upper bounded by $Ce^{-a \delta^2}$, where $C,a$ are some given constants independent from $\delta$. 

Consider the total quantum state space, it has real dimension $d_R=2d_Ad_B-1$. According to Lemma~\ref{avg-pur}, one choice of the coefficients are $C=2$, and $a=\frac{d_R+1}{9\pi^3 L^2\log 2 }\geq \frac{2d_Ad_B}{144\pi^3\log 2}$, where we have proved that the Lipschitz constant $L\leq 4$ in Lemma~\ref{lipschitz-pur}. Then we come to the following probabilistic deviation bound:
\begin{equation}
\mathrm{Prob}_H \left[ \left| \mathrm{Tr}\rho_{A}^{2}-  \frac{d_{A}+d_{B}}{d_{A}d_{B}+1}\right| \geq \delta \right] \leq 2\exp \left( - \frac{d_{A}d_{B}{\delta}^{2}}{72\pi^{3}\log 2}\right).
\end{equation}
By Definition~\ref{def:k-u} the subsystem purity $\tr \rho_A^2\in [\frac 1 {d_A}, \frac 1 {d_A} + \epsilon^2]$. In our discussion we use the symbols $\delta,\Delta$ defined in Eq.~\eqref{def:delta-and-Delta} and impose the constraint $\epsilon^2\geq 2\Delta$ on $\epsilon$.
\begin{figure}
    \centering
    \includegraphics[width=\linewidth]{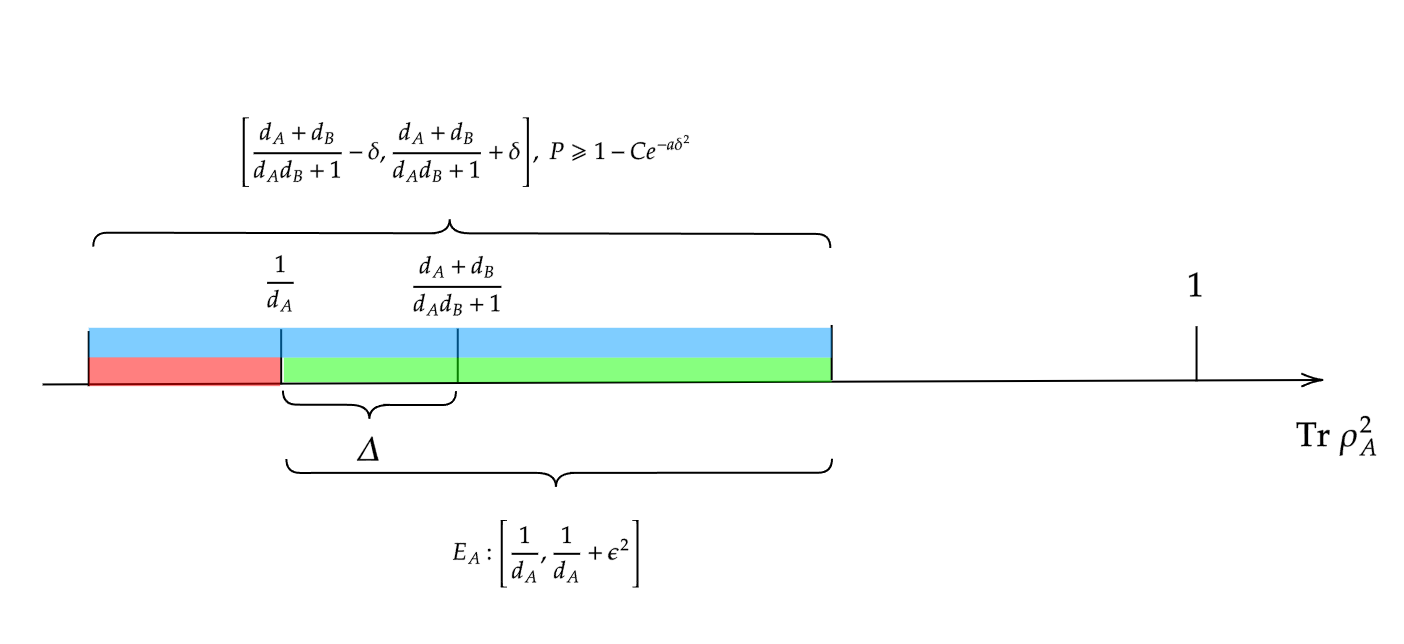}
    \caption{Description of Eq.~\eqref{eq:convert} when $\epsilon^2\geq 2\Delta$. Since $\tr \rho_A^2< \frac {1}{d_A}$ is impossible (red region), setting $\delta=\epsilon^2-\Delta$ converts the concentration results centering average value (blue region) to the probability lower bound of event $E_A$ (green region).}
    \label{fig:region}
\end{figure}
Since the purity can not be smaller than $\frac 1 {d_A}$, the two probabilities are equal (see FIG.~\ref{fig:region}), and the Haar random state satisfies the $\epsilon$-approximate $k$-uniform state condition on subsystem $A$ (denote this event as $E_{A}$) with probability
\begin{equation}
\mathrm{Prob}_H[E_{A}] = \mathrm{Prob} _H\left[ \left| \mathrm{Tr}\rho_{A}^{2}-  \frac{d_{A}+d_{B}}{d_{A}d_{B}+1}\right|\leq \delta \right] \geq 1-2 \exp \left( - \frac{d_{A}d_{B}\delta^{2}}{72\pi^{3}\log 2} \right).\label{eq:convert}
\end{equation}
Denote the event that a Haar random state satisfy the $\epsilon$-approximate $k$-uniform state condition on any $k$-partite subsystem as $E$. In the following discussion we use superscript $c$ to refer to the compliment of a event. By definition $E=\bigcap_{A\subseteq[n],|A|=k}E_{A}$. Using de-Morgan's rule and union theorem of probability measure,
\begin{equation}
\begin{aligned}
\mathrm{Prob}_H(E^{c}) & =\mathrm{Prob}_H\left( \bigcup_{A\subseteq[n],|A|=k} E_{A}^{c}\right)  \\
  & \leq \sum_{A\subseteq[n],|A|=k} \mathrm{Prob}_H\left( E_{A}^{c}\right) .
\end{aligned}
\end{equation} 
In our settings $d_{A}=d^{k},d_{B}=d^{n-k}$, so
\begin{equation}
\mathrm{Prob}_H[E_{A}^{c}]\leq 2\exp\left(- \frac{d^{n}\delta^{2}}{72\pi^{3}\log 2}\right),
\end{equation}
and we come to the final result
\begin{equation}
    \begin{aligned}
\mathrm{Prob}_H[E^{c}] & \leq {n\choose k} 2\exp\left(-\frac{d^{n}\delta^{2}}{72\pi^{3}\log2}\right),  \\
  \mathrm{Prob}_H[E] &\geq 1- {n\choose k } 2\exp \left(-\frac{d^{n}\delta^{2}}{72\pi^{3}\log2}\right) .
    \end{aligned}
\end{equation}
We have thus proved the probability lower bound for Haar random state being an approximate $k$-uniform state.
\end{proof}

\subsection{Random circuit construction of approximate $k$-uniform states}
From the discussions above, we have shown that it is simple to generate an $\epsilon$-approximate $k$-uniform state if we can generate Haar random unitary. However, Haar random unitary is too ideal since it requires a quantum circuit of depth exponential in the number of qubits, which is experimentally challenging for today's quantum experiments. To simulate the behavior of random sampling under Haar measure in a unitary group, the concept of $\epsilon$-approximate $t$-design was proposed. 

There are multiple forms of definitions to the approximate uniform states (see Ref.\cite{low2010pseudorandomnesslearningquantumcomputation} for different kinds of definitions). These definitions are useful in different scenarios. In this section, we will relate the definition based on relative error (used in Ref.~\cite{schuster2024random}) and the definition based on monomials (used in Ref.~\cite{Low_2009}).
\begin{definition}[$\epsilon$-approximate $t$-design based on relative error]
\label{approx-k-design-def1}
     An $\epsilon$-approximate unitary $t$-design is defined as the random unitary ensemble $\mathcal E$  such that the corresponding channel $\Phi_\mathcal{E}$ is $\epsilon$-close to the Haar random unitary channel $\Phi_H$:
\begin{equation}
(1-\epsilon) \Phi_H \preceq \Phi_\mathcal E \preceq (1+\epsilon) \Phi_H,
\end{equation}
where the quantum channel $\Phi_\mathcal E(\cdot)$ is defined as
\begin{equation}
\Phi_\mathcal{E}(A):=\mathbb E_{U\sim \mathcal E}\left[U^{\otimes t} A U^{\dagger \otimes t}\right],
\end{equation}
and $\Phi \preceq \Phi'$ denotes that $\Phi'-\Phi$ is a completely-positive map.
 \end{definition}

 In the paper presenting large deviation bounds \cite{Low_2009}, the definition based on monomials is used: we call the matrix function $M(U)$ a \emph{balanced monomial of degree $k$ }if it has the form $M(U)=U_{i_1j_1}U_{i_{2}j_{2}}\dots U_{i_{k}j_{k}}U^*_{i_{1}'j_{1}'}\dots U^{*}_{i_{k}'j_{k}'}$, where $i_{\alpha},j_{\alpha},i_{\alpha}',j_{\alpha}'\in [d](\alpha=1,2,\dots,k)$ are arbitrary indices. We have the following definition of $\epsilon$-approximate $k$-design:

\begin{definition}[$\epsilon$-approximate $t$-design based on monomials]
\label{approx-k-design-def2}
$\nu$ is an $\epsilon$-approximate unitary $t$-design if, for all balanced monomials $M$ of degree $\leq t$, 
\begin{equation}
\left| \mathbb{E}_{U\sim \nu} M(U) - \mathbb{E}_{U\sim H} M(U) \right| \leq \frac{\epsilon}{D^{t}},
\end{equation}
where $D$ is the total system dimension.
\end{definition}

These definitions are equivalent in the sense that if $\nu$ is an $\epsilon$-approximate $k$-design in Definition~\ref{approx-k-design-def1}, then it is an $\epsilon'$ approximate $k$-design in Definition~\ref{approx-k-design-def2} with $\epsilon'=\alpha\epsilon$, where $\alpha$ is related to system dimension. To find their equivalence factor, we integrate the results from Ref.~\cite{schuster2024random,low2010pseudorandomnesslearningquantumcomputation}, and proposed the following lemma connecting the multiplicative error and monomial-based error:
\begin{lemma}[Multiplicative and monomial-based error; see Lemma 4 in \cite{schuster2024random} and Lemma 2.2.14 in \cite{low2010pseudorandomnesslearningquantumcomputation}]
Any unitary $t$-design with relative error $\epsilon'$ (Definition~\ref{approx-k-design-def1}) is a unitary $t$-design with an additive error (under diamond norm) $2\epsilon'$; Any unitary $t$ design with additive error (under diamond norm) $2\epsilon'$ is a $2\epsilon' D^{\sigma t}$-approximate unitary $t$-design based on monomials (Definition~\ref{approx-k-design-def2}), where $\sigma$ is a constant that can be taken as $\frac 72$.
\end{lemma}

There have been many remarkable developments\cite{schuster2024random,Low_2009,Zhou_2022,Zhu_2017,zhu2016cliffordgroupfailsgracefully} on the technique of generating $\epsilon$-approximate $k$-design ensembles: for example, Ref.~\cite{Zhou_2022} investigated the properties of random hypergraph states generated by random controlled-Z and controlled-controlled-Z gates; Ref.~\cite{schuster2024random} proved that logarithmic-depth circuit is enough to generate an $\epsilon'$-approximate $t$-design with large $t$: more specifically,
\begin{theorem}[Generating $\epsilon'$-approximate $t$-designs in low depth; see~\cite{schuster2024random}]\label{approx-t-design-in-low-depth}
An $\epsilon'$-approximate $t$-design (based on Definition~\ref{approx-k-design-def1}) can be generated in circuit depth 
\begin{itemize}
    \item $L=\mathcal{O}(\log(n/\epsilon')\cdot t\operatorname{polylog}(t))$ for 1D circuits without ancilla qubits, or
    \item $L=\mathcal O(\log\log (n/\epsilon'))$ for all-to-all circuits with $\mathcal O(n\log (n/\epsilon'))$ ancilla qubits and $t\leq 3$.
\end{itemize} 
\end{theorem}

Moreover, Ref.~\cite{Low_2009} presented concentration results on pseudo-random distribution, which is a similar result to Levy's lemma, but has weaker constraint on the underlying probability measure in the sense that the ensemble only forms an $\epsilon'$-approximate $t$-design instead of an ideal Haar measure. This makes the corresponding results applicable in shallow random quantum circuits.
\begin{lemma}[Deviation bounds on $\epsilon'$-approximate $t$-design; see \cite{Low_2009}\label{thm:prob-concen-k-design}]

Let $f$ be a polynomial of degree $K$. Let $f(U)=\sum_{i}\alpha_{i}M_{i}(U)$ where $M_{i}(U)$ are monomials and let $\beta(f)=\sum_{i}|\alpha_{i}|$. Suppose $f$ has probability concentration under Haar measure:
\begin{equation}
\mathrm{Prob}_H(|f-\mu|\geq\delta) \leq Ce^{-a\delta ^{2}},
\end{equation}
where $C,a$ are some given constants. Let $\nu$ be an $\epsilon'$-approximate unitary $t$-design (Definition~\ref{approx-k-design-def2} based on monomials). Then
\begin{equation}
\mathrm{Prob}_{\nu}(|f-\mu|\geq\delta) \leq \frac{1}{\delta ^{2m}} \left( C \left( \frac{m}{a} \right)^{m} + \frac{\epsilon'}{D^{t}}(\beta(f) + |\mu|)^{2m} \right)
\end{equation}
for integer $m$ with $2mK\leq t$, where $D$ is the total system dimension, and $\mathrm{Prob}_H,\mathrm{Prob}_\nu$ represents the probability under Haar measure and approximate $t$ design measure, respectively.
\end{lemma}

It seems that it only requires the ensemble to be a $2$-design to generate approximate $k$-uniform states, since approximate $k$-uniform only has purity constraints on reduced density matrices. However, we find that for pseudo-random unitary ensemble, it requires $t$ to be bigger to reach similar concentration results from Haar random ensemble (Theorem~\ref{thm:approx-uniform-prob}), but a constant $t$ suffices for our purpose.

In this section, we integrate these results and show that low depth circuit can generate approximate $k$-uniform states with high probability. We come to the following conclusion:

\begin{theorem}\label{thm:kdesign-kuniform}
$\epsilon$-approximate $(\alpha n)$-uniform state on $n$ qudits with local dimension $d$, $\epsilon$ vanishingly small can be generated from one-dimensional(1D) random circuit ensemble of depth $L=O(n)$ with vanishingly small failure probability,
given the constraint that $\epsilon=\omega(d^{- (1-\alpha)n/2})\label{eq:orig-restr}$ and $\epsilon=\omega(d^{-n/4})$.
\end{theorem}

\begin{proof}

First for our purposes, the constants in Theorem~\ref{thm:prob-concen-k-design} are $a= \frac{d_{A}d_{B}}{72\pi^{3}\log 2}$, $C=2$, and $\mu= \frac{d_{A}+d_{B}}{d_{A}d_{B}+1}$. Next we calculate the explicit expression of purity function $f_{A}(\psi)=\mathrm{Tr}\rho _A^{2}$ in terms of polynomial in $U$: due to unitary invariance, suppose random state $\psi$ is evolved from $\ket{0}_{A}\ket{0}_{B}$ and the corresponding unitary $U$ in the ensemble has matrix element $U_{ab,a'b'}:=\bra a\bra b U \ket {a'}\ket{b'}$:

\begin{equation}
\ket{\psi}   =U\ket{0}_{A}\ket{0} _{B}=\sum_{a=1}^{d_{A}}\sum_{b=1}^{d_{B}}U_{ab,00}\ket{a} _{A}\ket{b}_{B}.
\end{equation}
Then the purity function can be expressed in terms of $U$:
\begin{equation}
\mathrm{Tr} \rho_{A}^{2} =  \sum_{bb',aa'} U_{a'b,00 }U^*_{ab,00}U_{ab',00}U^*_{a'b',00}.
\end{equation} 
It is direct to recognize that $\mathrm{Tr}\rho _A^{2}$ is a polynomial of $U$ with degree $K=2$ and $\beta(f)=d_{A}^{2}d_{B}^{2}$, then for $\epsilon'$-approximate $t$-design $\nu$, the concentration result gives

\begin{equation}
\mathrm{Prob}_{\nu}\left( \left|f_{A}(\psi)- \frac{d_{A}+d_{B}}{d_{A}d_{B}+1}\right| \geq \delta \right) \leq \delta ^{-2m}\left( 2 \left( \frac{72\pi^{3}\log2}{d_{A}d_{B}}m \right)^{m} + \frac{\epsilon'}{D^{t}} \left( d_{A}^{2}d_{B}^{2} + \frac{d_{A}+d_{B}}{d_{A}d_{B}+1} \right)^{2m} \right)
\end{equation}
with integer $4m\leq t$ and total dimension $D=d^n$. We may set this to the largest $m=\lfloor \frac{t}{4} \rfloor$ for optimal result. Similar to the discussion in the Haar random case, we use the symbol $\Delta$ defined in Eq.~\eqref{def:delta-and-Delta} and impose the constraint $\epsilon^2\geq 2\Delta$ on $\epsilon$.
Then it satisfies the $\epsilon$-approximate $k$-uniform state constraint on $k$-partite subsystem $A$ with probability greater than
\begin{equation}
\mathrm{Prob}_\nu[E_{A}]\geq 1- \left(\frac{1}{d_{A}} + \epsilon^2 - \frac{d_{A}+d_{B}}{d_{A}d_{B}+1}\right)^{-2m}\left( 2 \left( \frac{72\pi^{3}\log2}{d_{A}d_{B}}m \right)^{m} + \frac{\epsilon'}{D^{t}} \left( d_{A}^{2}d_{B}^{2} + \frac{d_{A}+d_{B}}{d_{A}d_{B}+1} \right)^{2m} \right),
\end{equation}
where $m=\lfloor \frac{t}{4} \rfloor$. For the case of satisfying all possible $k$ parties. Using similar techniques to the proof above, we find that the failure probability is lower bounded by
\begin{equation}
\mathrm{Prob_\nu(fail)}  \leq {n\choose k} \left(\frac{1}{d_{A}} + \epsilon^2 - \frac{d_{A}+d_{B}}{d_{A}d_{B}+1}\right)^{-2m}\left( 2 \left( \frac{72\pi^{3}\log2}{d_{A}d_{B}}m \right)^{m} + \frac{\epsilon'}{D^{t}} \left( D^{2} + \frac{d_{A}+d_{B}}{d_{A}d_{B}+1} \right)^{2m} \right).\label{eq:probfail}
\end{equation}
Note that we have been using the Definition~\ref{approx-k-design-def2} of approximate $t$-designs. Here we convert the definition to Definition~\ref{approx-k-design-def1} by simply making the replacement $\epsilon'\to 2\epsilon'D^{\sigma t}$ where $D=d^n$ and $\sigma=\frac 72$, so that we can relate the $\epsilon'$ and $t$ to circuit depth $L$ using Theorem~\ref{approx-t-design-in-low-depth}.
\begin{equation}
\mathrm{Prob_\nu(fail)}  \leq {n\choose k} \left(\frac{1}{d_{A}} + \epsilon^2 - \frac{d_{A}+d_{B}}{d_{A}d_{B}+1}\right)^{-2m}\left( 2 \left( \frac{72\pi^{3}\log2}{d_{A}d_{B}}m \right)^{m} + \frac{2\epsilon'D^{\sigma t}}{D^{t}} \left( D^{2} + \frac{d_{A}+d_{B}}{d_{A}d_{B}+1} \right)^{2m} \right).
\end{equation}

Denote error magnitude $\gamma:=-\log \mathrm{Prob}(\mathrm{fail}) $. Taking negative logarithm on both sides of Eq.~\eqref{eq:probfail}, that is
\begin{equation}
\begin{aligned}
\gamma  & \geq -nS(\alpha)-\log(n+1)+1 +2m \log \left(\frac 1{d_A} + \epsilon^2 - \frac{d_{A}+d_{B}}{d_{A}d_{B}+1} \right) - \log \left[ 2 \left( \frac{72\pi^{3}\log 2}{D}m \right)^{m} + 2\frac{\epsilon'D^{\sigma t}}{D^{t}} \left( D^{2} + \frac{d_{A}+d_{B}}{d_{A}d_{B}+1} \right)^{2m} \right]\\
  &  \simeq -n S(\alpha) + t\log\epsilon - \max  \left( \frac{t}{4} (\log t - n \log d), \sigma tn\log d - \log  \frac{1}{\epsilon'} \right) \\
  & \simeq  -n S(\alpha) + t\log \epsilon + \frac{n}{4}t\log d -\max \left( 0, -\log \frac{1}{\epsilon'}+  \left( \sigma + \frac 14 \right) nt \log d \right).\label{derivation-approx-design-k-uniform}\\
\end{aligned}
\end{equation}
In the first line of Eq.~\eqref{derivation-approx-design-k-uniform}, we have used Eq.~\eqref{eq:log-comb-number} for bounding logarithm of combinatorial numbers. In the second line of Eq.~\eqref{derivation-approx-design-k-uniform}, to do asymptotic analysis, we have taken $d_{A}=d^{k}$, $d_{B}=d^{n-k}$, $D=d_Ad_B=d^n$, $m=\frac{t}{4}$, $\log(a+b)\simeq\max(\log a,\log b)$ holds asymptotically, assumed $\epsilon=\omega(d^{- n(1-\alpha)/2})$ and removed all non-leading terms; for example, we have $\frac 1 {d_A}+\epsilon^2 - \frac{d_A+d_B}{d_Ad_B+1}\simeq \epsilon^2$, $D^2+ \frac{d_A+d_B}{d_Ad_B+1}\simeq D^2$, and $\log\left[ (\frac c Dm)^m\right]=\frac t4\log(\frac c {d^n}\frac t4)\simeq \frac t4(\log t - n\log d)$. In the third line of Eq.~\eqref{derivation-approx-design-k-uniform} we used the assumption that $\log t=o(n)$, so $t(\log t-n\log d)\simeq -tn\log d$, otherwise the random circuit has exponential depth, which is capable of generating exact Haar random unitary. Then we take $-\frac{tn}4 \log d$ out of maximization for clarity.

To determine which term has larger magnitude, we have to discuss the magnitude of $\log \frac{1}{ \epsilon'}$ and $t$ case by case:
\begin{itemize}
    \item $\log \frac{1}{\epsilon'} \geq (\sigma+\frac 14)nt \log d$, in this case we have $\gamma\geq -nS(\alpha) + t\log \epsilon  + \frac{n}{4}t\log d$, we can take the $\log \frac{1}{\epsilon'}$ to be $\Theta(nt)$ to have optimal circuit depth $L=O(nt^{2}\operatorname{polylog} t)$, next we can find the minimum possible $t$ by letting $\gamma$ to grow linearly in $n$: let $S_{d}(\alpha)= \frac{S(\alpha)}{\log d}$ to be the entropy taken in base $d$, and $\lambda : = \frac{1}{n}\log_{d}\epsilon$ to be the error rate. We have that in the large-$n$ limit,
\begin{equation}
\gamma \gtrsim n \log d\left(\left(\frac{1}{4} + \lambda \right)t- S_{d}(\alpha)\right),
\end{equation}
to have vanishingly small failure probability we only require that $t>\frac{S_{d}(\alpha)}{1 / 4 + \lambda  }= O(1)$ and $\lambda > -\frac 14$ (or $\epsilon=\omega(d^{-n/4})$ equivalently). And thus the corresponding circuit depth requirement is $L=O(n)$. This completes the proof that linear depth circuit is capable of generating $\epsilon$-approximate $k$-uniform states with $\epsilon$ vanishingly small, $k=\Theta(n)$ and failure probability vanishingly small.
\item  $\log \frac{1}{\epsilon'} \leq (\sigma +\frac 14) nt\log d$, in this case we have $\gamma\geq -nS(\alpha)+ t\log\epsilon + \frac{n}{4}t\log d - ( \left( \sigma + \frac 14 \right) nt\log d - \log \frac{1}{\epsilon'})$. This can be rewritten as
\begin{equation}
\gamma\gtrsim n\log d \left( -S_{d}(\alpha) + \lambda t  + \frac{1}{n}\log_{d} \frac{1}{\epsilon'} -  \sigma t\right).
\end{equation}
Since $\lambda<0$, we can see that to keep the right hand side of inequality to be positive, we have to set $ \log_{d} \frac{1}{\epsilon'}>  n((\sigma -\lambda) t+S_{d}(\alpha))$, note that such $\epsilon'$ can be acquired only when $t> S_d(\alpha)/(\lambda + \frac 14)$ and $\lambda > -\frac 14$. In such conditions we can also construct high-performance $k$-uniform states within linear circuit depth.

\end{itemize}

\end{proof}

\subsection{Approximate QECCs}
In this section we propose a definition of approximate QECCs based on approximate $k$-uniform states. 

We first introduce the notion of QECC. Let $\{e_j\}_{i\in \bbZ_{d^2}}$ ($e_0=I_d$) be an orthogonal basis on linear operators acting on $\bbC^{d}$, i.e., $\tr(e_i^{\dagger}e_j)=\delta_{i,j}d$. Then  
$\cE:= \{E_{\alpha}=e_{\alpha_1}\otimes e_{\alpha_2}\otimes\cdots\otimes e_{\alpha_n}\mid \alpha=(\alpha_1,\alpha_2,\ldots,\alpha_n)\in \bbZ_{d^2}^n\}$ is the error ensemble on $(\bbC^{d})^{\otimes n}$. Every error $E_\alpha\in\mathcal E$ has norm $||E_\alpha||=\sqrt{\tr (E_\alpha^\dagger E_\alpha})=\sqrt{d^n}$. For an error $E_\alpha \in \cE$, the support of $E_\alpha$ is defined as $\mathrm{supp}(E_\alpha):=\{i\mid e_{\alpha_i}\neq I_d,\, 1\leq i\leq n\}$, and the weight of $E_\alpha$ is defined as $\mathrm{wt}(E_\alpha):=|\mathrm{supp}(E_\alpha)|$. 

Since $\mathcal E$ forms a complete basis of linear operators acting on the space, we can expand any linear error operator $E\in \mathcal L((\mathbb C^d)^{\otimes n})$ as the linear combination of $E_\alpha\in\mathcal E$: $E=\sum_{\alpha\in \mathbb Z_{d^2}^2}c_\alpha E_\alpha$, where $c_\alpha\in\mathbb C$ are the coefficients. We define the support of a linear operator $E$ to be the union of the support of all its expansion basis with non-zero coefficients:
$\mathrm{supp}(E):=\bigcup_{\alpha,c_\alpha\neq 0}\supp(E_\alpha)$, and the weight of $E$ being $\mathrm{wt}(E):=|\supp(E)|$.
A QECC on $n$ qudits with local dimension $d$ and code distance $\delta$, denoted as $((n,K,\delta))_d$, is a $K$-dimensional subspace $\cal C$ of $(\bbC^{d})^{\otimes n}$ satisfying that: for all states $\ket{\psi}\in \cal C$, and all errors $E_\alpha\in \cal E$ with $\mathrm{wt}(E_\alpha)<\delta$,
\begin{equation}
 \bra\psi E_\alpha \ket\psi = C(E_\alpha),   
\end{equation}
where $C(E_\alpha)$ is a constant independent of $\ket{\psi}$. Moreover,  a QECC is called \emph{pure} if $C(E_\alpha)=\frac{\tr(E_\alpha)}{d^n}$. Note that, a subspace $\cC$ of $(\bbC^d)^{\otimes n}$ is a pure $((n,K,\delta))_d$ QECC if and only if each pure state of $\cC$ is a $(\delta-1)$-uniform state \cite{PhysRevA.104.032601,Huber_2020}.

The condition for perfect error correction is that the entanglement of logical qubits with other system can be restored perfectly. The approximate error correction (AQEC) can be considered in this framework as the entanglement entropy between reference system and environment being small \cite{schumacher2001approximatequantumerrorcorrection}. There are many works on understanding the properties of AQEC \cite{Leung_1997,Liu_2023,Yi_2024,liu2025approximatequantumerrorcorrection}.  In this work, relating to approximate $k$-uniform states, we give the following definition for approximate QECCs.

\begin{definition}\label{thm:pure-aqec}
A $K$-dimensional subspace $\mathcal{C}$
of $(\bbC^d)^{\otimes n}$ is called an $\epsilon$-approximate $((n,K,\delta))_d$ QECC if for all states $\ket{\psi}\in \mathcal{C}$, and for any linear error operator $E$ with  $\mathrm{wt}(E)<\delta$ and norm $||E||=\sqrt {d^n}$ (suppose $S= \mathrm{supp}(E)$),
 we have
\begin{equation}
\left| \left<\psi|E|\psi\right> - \mathrm{Tr}\left( E \frac{P}{K}\right) \right| \leq \epsilon\sqrt{ d_S }.
\end{equation}
Furthermore, the  $\epsilon$-approximate $((n,K,\delta))_d$ QECC is called pure if
\begin{equation}
\left| \left<\psi|E|\psi\right> - \frac{1}{d_S}\mathrm{Tr}(E_S) \right| \leq \epsilon\sqrt{ d_S },
\end{equation}
where $E_S$ is the part of the operator that acts on subsystem S: $E=E_S\otimes I_{S^c}$.
\end{definition}
When $\epsilon=0$, an $\epsilon$-approximate (respectively, pure) $((n,K,\delta))_d$ QECC must be an exact (respectively, pure)  $((n,K,\delta))_d$ QECC.
Note that one-dimensional $\epsilon$-approximate $((n,1,\delta))_d$   QECCs refer only to $\epsilon$-approximate pure $((n,1,\delta))_d$ QECCs. Based on the connection  between QECCs and $k$-uniform states \cite{PhysRevA.104.032601,Huber_2020}, we also build the connection between approximate QECCs and approximate $k$-uniform states.

\begin{theorem}\label{thm:aqecc-eq-uniform}
A subspace $\cC$ of $(\bbC^{d})^{\otimes n}$ is an $\epsilon$-approximate pure $((n,K,\delta))_{d}$ QECC if and only if each pure state of $\cC$ is an $\epsilon$-approximate $(\delta-1)$-uniform state.  In particular, an $\epsilon$-approximate  pure $((n,1,k+1))_{d}$ QECC is equivalent to an $\epsilon$-approximate $k$-uniform state.
\end{theorem}

\begin{proof}
First, we prove the necessity.   By the definition of approximate pure QECC, each state $\ket{\psi}\in \mathcal{C}$ satisfies that for any linear error operator $E$ with $\mathrm{wt}(E)<\delta$ and norm $||E||= \sqrt {d^n}$, denote $S:= \mathrm{supp}(E)$, and we have
\begin{equation}\label{eq:psi_appro}
\left| \left<\psi|E|\psi\right> - \frac{1}{d_S}\mathrm{Tr}(E_S) \right| \leq \epsilon\sqrt{ d_S }.
\end{equation}
where $E=E_S\otimes I_{S^c}$ with norm $||E_S||=\sqrt {d_S}$. 
Suppose the reduction of $\ket{\psi}$ on subsystem $S$ can be represented as
\begin{equation}
\ket{\psi} \bra{\psi} _{S}= \frac{1}{d_S} (I_{S}+\tilde{P}_{S}),
\end{equation}
then
\begin{equation}\label{eq:psi_expre}
 \left<\psi|E|\psi\right> - \frac{1}{d_S}\mathrm{Tr}(E_S) = \mathrm{Tr}\left[ E_S\left( \ketbra{\psi}{\psi}_{S}- \frac{1}{d_S}I_{S} \right) \right] = \frac{1}{d_S} \mathrm{Tr} [E_S\tilde{P}_{S}].
\end{equation} 
According to Eq.~\eqref{eq:psi_appro} and  Eq.~\eqref{eq:psi_expre}, we have
\begin{equation}
\left| \mathrm{Tr} [E_S\tilde{P}_{S}] \right|  \leq \epsilon \sqrt{ d_S^3}.
\end{equation}
Note that this is true for \emph{any} $E$ with support on $S$. Since $|\tr[E_S \tilde P_S]|\leq ||E_S|| \cdot ||\tilde P_S||=\sqrt{d_S} ||\tilde P_S||$ with equal sign gained when $\tilde P_S=kE_S^\dagger$ for some constant $k$.  So we have $||\tilde P_S||\leq d_S\epsilon$, and from our equivalent Definition~\ref{alternate-def2} of $\epsilon$-approximate $k$-uniform state, $\ket{\psi}$ is an $\epsilon$-approximate $(\delta-1)$-uniform state. The sufficiency can be easily derived by reversing the above proof.
\end{proof}
\vspace{0.4CM}

To detect the distance of a QECC, it is convenient to use quantum weight enumerators. The\emph{ Shor-Laflamme weight enumerators} are first introduced in 
\cite{shor1996quantummacwilliamsidentities}:
\begin{equation}
    \begin{aligned}
A_j(M_1,M_2) &:= \sum_{\substack{E_\alpha\in \mathcal E\\\mathrm{wt}(E_\alpha)=j}} \mathrm{Tr}(E_\alpha M_1)\mathrm {Tr}(E_\alpha^\dagger M_2);\\
B_j(M_1,M_2) &:= \sum_{\substack{E_\alpha\in \mathcal E\\\mathrm{wt}(E_\alpha)=j}} \mathrm{Tr}(E_\alpha M_1E^{\dagger}_\alpha M_2).        
    \end{aligned}
\end{equation}
Later, Rains \cite{796376,rains1996quantumweightenumerators} introduced the \emph{unitary weight enumerators}:
\begin{equation}
    \begin{aligned}
            A_j'(M_1,M_2):=\sum_{\substack{T\subseteq [n]\\|T|=j}}A_T(M_1,M_2):= \sum_{\substack{T\subseteq [n]\\|T|=j}}\tr[(M_{1})_{T}(M_{2})_{T}];\\
    B_j'(M_1,M_2):=\sum_{\substack{T\subseteq [n]\\|T|=j}}B_T(M_1,M_2):=\sum_{\substack{T\subseteq [n]\\|T|=j}}\tr[(M_{1})_{T^c}(M_{2})_{T^c}].
    \end{aligned}
\end{equation}
These unitary weight enumerators are linearly related to Shor-Laflamme enumerators and both can be used to characterize the distance of QECCs \cite{rains1996quantumweightenumerators}. More specifically, suppose $P$ is the projector of a $K$-dimensional QECC, and $\tilde P=\frac{P}{K}$ is the normalized projector. Then the QECC is a $((n,K,\delta))_d$ QECC if and only if $KB_j(\tilde P,\tilde P)-A_j(\tilde P,\tilde P)=0$ for $0<j<\delta$, or $KB_{\delta-1}'(\tilde P,\tilde P)-A'_{\delta-1}(\tilde P,\tilde P)=0$, and the code is pure if and only if $B_j(\tilde P,\tilde P)=A_j(\tilde P,\tilde P)=0$ for $0<j<\delta$, or $KB_{\delta-1}'(\tilde P,\tilde P)=A_{\delta-1}'(\tilde P,\tilde P)={n\choose \delta-1} d^{1-\delta}$. Moreover, these weight enumerators can be evaluated efficiently from some specially designed quantum circuits \cite{miller2024experimentalmeasurementphysicalinterpretation,shi2024exploringquantumweightenumerators}.

Next we develop bounds of the quantum weight enumerators for approximate QECCs: we will use the Rains unitary enumerators $A_{S}'$ and $B_{S}'$. It is already known that $KB'_i(P,P)-A_i'(P,P)\geq 0$ for any projector $P$, and $KB_{\delta-1}’(P,P)-A_{\delta-1}'(P,P)=0$ for exact $((n,K,\delta))_d$ QECC. We can also give an upper bound of this quantity for the approximate case. Now consider the code space average $\mathbb{E}_{v\in \mathcal{C}} |\mathrm{Tr}[ (U_{S}\otimes I_{S^c})vv^{\dagger}]- \mathrm{Tr} [(U_{S}\otimes I_{S^c}) \frac{P}{K}]|^{2}$, where $P$ is the projector onto code space $\mathcal C$ with dimension $K$, $v\in \mathcal{C}$ is a randomly selected vector in code space, and $U_S$ is a unitary operator on subsystem $S\subseteq [n]$. Suppose $O$ be any operator, we have the following relation \cite{796376}:
\begin{equation}
\begin{aligned}
\mathbb{E}_{v\in \mathcal{C}}\left[ \left|\mathrm{Tr} (Ovv^{\dagger}) - \frac{1}{K} \mathrm{Tr}(OP)\right|^{2} \right]  & = \mathbb{E}_{v\in \mathcal{C}}(|\left<v| O | v\right>|^{2}) - \frac{1}{K^{2}} |\mathrm{Tr} (OP)|^{2}  \\
  & = \frac{1}{K(K+1)} [|\mathrm{Tr}(OP)|^{2} + \mathrm{Tr}(OPO^{\dagger}P)] - \frac{1}{K^{2}}|\mathrm{Tr}(OP)|^{2} \\
  & = \frac{1}{K^{2}(K+1)} (K\mathrm{Tr} (OPO^{\dagger}P)- \mathrm{Tr}(OP)\mathrm{Tr}(O^{\dagger}P)).    
\end{aligned}
\end{equation}
Replacing the operator $O$ with a unitary $U^S:=U_{S}\otimes I_{S^c}$, where $U_S$ is a unitary acting on subsystem $S$,
\begin{equation}
\mathbb{E}_{v\in \mathcal{C}} \left|\mathrm{Tr} (U^Svv^{\dagger})- \mathrm{Tr}\left( U^S \frac{P}{K}\right)\right|^{2}=\frac{1}{K^{2}(K+1)} \left(K \mathrm{Tr} (U^SPU^{S\dagger}P) - \mathrm{Tr} \left(U^SP\right) \mathrm{Tr} \left(U^{S\dagger}P\right)\right).
\end{equation}
Taking average over $U_{S}$ on both sides, using the positivity of left hand side, we have thus showed that $KB_{S}'(P,P)-A_{S}'(P,P)\ge0$:
\begin{equation}
\mathbb{E}_{v\in \mathcal C}\mathbb{E}_{U^{S}}\left|\mathrm{Tr} (U^{S}vv^{\dagger})- \mathrm{Tr} \left(U^{S} \frac{P}{K}\right)\right|^{2} = \frac{1}{K^{2}(K+1)} (KB_{S}'(P,P)-A_{S}'(P,P))\geq 0.\label{eq:avg-WENUM}
\end{equation}
 Using Eq.~\eqref{eq:avg-WENUM}, it is possible to give an upper bound on $KB'_S(P,P)-A'_S(P,P)$ in the approximate QECCs.
\begin{theorem}\label{thm:enumerator-bound-aqecc}
    If the subspace $\mathcal{C}$ is an $\epsilon$-approximate  $((n,K,\delta))_{d}$ QECC with  projector $P$, then for any subsystem $S$ with $|S|<\delta$, the quantity $KB_{S}'(P,P)-A_{S}'(P,P)$ is upper bounded by
\begin{equation}
KB_{S}'(P,P)-A_{S}'(P,P)\leq K^{2}(K+1) \epsilon ^{2}{ d_S },
\end{equation}
and also for any $i<\delta$,
\begin{equation}
KB_{i}'(P,P)-A_{i}'(P,P) \leq {n\choose i} K^{2}(K+1)\epsilon ^{2}d^{i}.
\end{equation}
\end{theorem}

\begin{proof}
According to Definition~\ref{thm:pure-aqec}, if $\mathcal C$ is an $\epsilon$-approximate $((n,K,\delta))_d$ QECC, and $P$ is a projector onto code space $\mathcal C$ and $v\in \mathcal C$ is a codeword, $\left| \mathrm{Tr} (E v v^{\dagger}) - \frac{1}{K} \tr (E P) \right|^{2}$ is upper bounded by $\epsilon ^{2}d_S$ as long as $\mathrm{wt}(E)<\delta$. Any unitary operator $U$ with support $\mathrm{supp}(U)\subseteq S$ satisfies this condition. Taking Haar average over error operators ranging unitary group of subsystem $S$, and code space $\mathcal{C}$, and using  Eq.~\eqref{eq:avg-WENUM}, we get the upper bound for $KB_{S}'(P,P)-A_{S}'(P,P)$:
\begin{equation}
KB_{S}'(P,P) - A_{S}'(P,P) =K^2(K+1)\mathbb E_{v\in\mathcal C}\mathbb E_{U^S}\left|\tr(U^Svv^\dagger)-\tr\left(U^S \frac PK\right)\right|^2\leq K^{2}(K+1)\epsilon ^{2}d_S.
\end{equation}
The $KB_{i}'(P,P)-A_{i}'(P,P)$ corresponds to the summation over all possible subsystem with size $i$, which is upper bounded by ${n\choose i}$ times the upper bound of every single $KB_S'(P,P)-A_S'(P,P)$ with $|S|=i$:
\begin{equation}
KB_{i}'(P,P)-A_{i}'(P,P) \leq {n\choose i} K^{2}(K+1)\epsilon ^{2}d^{i}.
\end{equation}
\end{proof}

If the approximate QECC is pure, it is possible to give bounds on weight enumerators respectively:
\begin{theorem}
    For a pure $\epsilon$-approximate $((n,K,\delta))_{d}$  QECC with projector $P$, if $|S|<\delta$, then
\begin{equation}
\frac{K^{2}}{d_S} \leq A_{S}'(P,P)\leq \frac{K^{2}}{d_S} + K^2\epsilon ^{2},
\end{equation}
and 
\begin{equation}
\frac{K}{d_S}\leq B_{S}'(P,P) \leq K\left[ \frac{1}{d_S} + \epsilon ^{2}\left( 1+ ({K+1}){d_S} \right) \right].
\end{equation}
\end{theorem}
\begin{proof}
Select a set of orthonormal states $\ket{\psi_i}(i=1,\cdots,K)$ in code space $\mathcal C$. By definition, the projector $P$ can be expressed as
\begin{equation}
    P=\sum_{i=1}^{K} \ket{\psi_i}\bra{\psi_i}.
\end{equation}
Using Theorem~\ref{thm:aqecc-eq-uniform}, all $\ket{\psi_i}$ are approximate $\delta-1$-uniform states, and thus the reduction on $S$ can be expressed as $\ket{\psi_i}\bra{\psi_i}_S=\frac{I_S+(\tilde P_{i})_S}{d_S}$ with $||(\tilde{P}_{i})_S||\leq d_S \epsilon$, then the reduction on subsystem $S$ of projector $P$ can be written as
\begin{equation}
    \begin{aligned}
P_{S} & =\sum_{i=1}^{K} \frac{1}{d_S} (I_{S} + (\tilde{P}_{i})_S)  \\
  & = \frac{K}{d_S} (I_{S}+ \tilde{P}_{S}),       
    \end{aligned}
\end{equation}
where $\tilde P_S:=\sum_{i=1}^{K} \frac{1}{K} (\tilde{P}_{i})_S$. Note that $|| \tilde{P}_{S}| |\leq \frac 1 K \sum_{i=1}^{K}||(\tilde{P}_{i})_S||\leq d_S\epsilon$, and $\tr \tilde P_S=0$. Then 
\begin{equation}
\frac{K^{2}}{d_S}\leq A_{S}'(P,P)=\mathrm{Tr} P_{S}^{2} = \frac{K^{2}}{d_S^2}\mathrm{Tr} (I_{S}+\tilde{P}_{S}^{2})\leq \frac{K^{2}}{d_S} + K^{2}\epsilon ^{2}.
\end{equation}
Using inequality in Theorem~\ref{thm:enumerator-bound-aqecc}, and $KB_S'(P,P)-A_S'(P,P)\geq 0$, we complete the proof:
\begin{equation}
\frac{K^2}{d_S}\leq A_S'(P,P)\leq KB_{S}'(P,P)\leq A_{S}'(P,P)+ K^{2}(K+1)\epsilon ^{2}d_S\leq K^{2}\left[ \frac{1}{d_S} + \epsilon ^{2}\left( 1 + (K+1)d_S  \right)\right].
\end{equation}    
\end{proof}
\vspace{0,4cm}

From Section~\ref{sec:haar-approximate-uniform} we see that a Haar random state is an approximate $k$-uniform state with high probability. Theorem~\ref{thm:aqecc-eq-uniform} proves the relation between approximate $k$-uniform state and approximate QECC. It is tempting to think that a Haar random subspace is an approximate QECC. We can clarify this point by calculating the averaged weight enumerators on random subspace: the Haar expectation of $\frac{KB'_S(P_U,P_U)-A_S'(P_U,P_U)}{K^2(K+1)}$ is small, where $P_U:= UP_0U^{\dagger}$ is a $K$-dimensional random subspace projector:
\begin{equation}
\mathbb E_{U}\left[\frac{KB'_S(P_U,P_U)-A_S'(P_U,P_U)}{K^2(K+1)}\right]=\left(1-\frac 1K\right) \frac{d_S-\frac 1 {d_S}}{d^n-\frac {1}{d^n}}.
\end{equation}
The right hand side is asymptotically $(1-\frac 1K) d^{|S|-n}$, and vanishingly small in the large $n$ limit. This implies that a ``good" approximate QECC can probably be constructed from Haar random unitaries, or even low depth random quantum circuits, with small failure probability. Next we show that ``good" approximate QECC can be constructed from a Haar random subspace.

\subsection{Haar random construction of approximate QECCs}\label{sec:Haar-random-aqecc}
In this section, we show the random construction of approximate QECC, and prove that a random subspace is an $\epsilon$-approximate pure $((n,K,\delta))_d$ QECC with high probability under given constraints. To prove this, we need conclusion from Theorem~\ref{thm:aqecc-eq-uniform} that a subspace is an approximate QECC if and only if all the states within the subspace is an approximate uniform state. We also need to use definition and property \cite{Hayden_2006,Hayden_2004} of an $\epsilon'$-net of the $K$-dimensional subspace. The following lemma shows the definition of $\epsilon'$-net, and the upper bound for the cardinality of smallest $\epsilon'$-net:

\begin{lemma}[Existence of small $\epsilon$-nets; See Lemma III.6 in \cite{Hayden_2006} or Lemma II.4 in \cite{Hayden_2004}]
    For $0<\epsilon<1$ and $\mathrm{dim} \mathcal H=K$, there exists a set $\mathcal N$ of pure states in $\mathcal H$ with cardinality $|\mathcal N|\leq (5/\epsilon)^{2K}$, such that for every pure state $\ket\varphi\in\mathcal H$ there exists $\ket{\tilde\varphi}\in \mathcal N$ such that $||\ket\varphi - \ket{\tilde\varphi}||_2\leq \epsilon/2$ and $||\ket\varphi - \ket{\tilde\varphi}||_1\leq \epsilon$.\label{lemma:epsilon-nets}
\end{lemma}

Next we propose an alternative definition of approximate pure QECC based on Theorem~\ref{thm:aqecc-eq-uniform}:
\begin{lemma}[Alternative definition of approximate pure QECC]\label{lemma:altdef-aqecc}
    A $K$-dimensional subspace $\mathcal C$ is the code space of an $\epsilon$-approximate pure $((n,K,\delta))_d$ QECC if and only if
\begin{equation}
    \max_{\ket\psi\in\mathcal C} \tr (\ket\psi\bra\psi_S^2 ) \leq \frac 1 {d_S} + \epsilon^2\label{eq:altdef-aqecc}
\end{equation}
is satisfied for any subsystem $S\subseteq [n]$ with $|S|=\delta-1$.
\end{lemma}
\begin{proof}
    We first prove the necessity. Suppose the $\mathcal C$ is an $\epsilon$-approximate $((n,K,\delta))_d$ pure QECC, then according to Theorem~\ref{thm:aqecc-eq-uniform}, any state $\ket\psi\in\mathcal C$ is an $\epsilon$-approximate $\delta-1$-uniform state. According to Definition~\ref{def:k-u}, for any $S\subseteq[n]$ and $|S|=\delta-1$,
    \begin{equation}
    \tr(\ket\psi\bra\psi_S^2)\leq \frac 1{d_S} + \epsilon^2    .
    \end{equation}
    Note that this holds true for any state $\ket\psi\in\mathcal C$, so by taking maximization over code space $\mathcal C$, we see that Eq.~\eqref{eq:altdef-aqecc} is true for any $S\subseteq[n]$ and $|S|=\delta-1$.

    Next we prove the sufficiency. Suppose Eq.~\eqref{eq:altdef-aqecc} holds true for any $S\subseteq[n]$ and $|S|=\delta-1$, for any state $\ket\psi\in \mathcal C$, $\tr(\ket\psi\bra\psi_S^2)\leq \max_{\ket{\psi'}\in\mathcal C}\tr(\ket{\psi'}\bra{\psi'}_S^2)\leq \frac 1 {d_S}+\epsilon^2$, so any state within the subspace is an $\epsilon$-approximate $\delta-1$-uniform state. By Theorem~\ref{thm:aqecc-eq-uniform} $\mathcal C$ is an $\epsilon$-approximate $((n,K,\delta))_d$ pure QECC.
\end{proof}

Using lemma~\ref{lemma:altdef-aqecc}, we can prove that a Haar random subspace is an approximate QECC with high probability:
\begin{theorem}
    A randomly selected $K$-dimensional subspace $\mathcal C\in_R\mathcal G_K(\mathcal H)$ is an $\epsilon$-approximate pure $((n,K,\delta))_d$ QECC with probability greater than 
    \begin{equation}
        1- \left(\frac 5{\epsilon'}\right)^{2K} {n\choose{\delta-1}} 2\exp\left(- \frac{d^n\mu^2}{72\pi^3\log 2}\right)
    \end{equation}
    for any $0<\epsilon'<1$, where $\mu=\epsilon^2-2\epsilon'-\Delta$, $\Delta:=\frac{d^{\delta-1}+ d^{n-(\delta-1)}}{d^{n}+1}- \frac{1}{d^{\delta-1}}$, given the constraint that $\epsilon^2-2\epsilon'\geq 2\Delta$.
\end{theorem}
We first do some asymptotic analysis as usual. Taking the negative logarithm to the failure probability, we have

\begin{equation}
-\log\text{Prob}(\mathrm{fail}) \geq -\log2 + 1 - \log(n+1) - nS(\alpha) - 2K\log \frac{5}{\epsilon'} +  \frac{d^{n}}{72\pi^{3}\ln2} (\epsilon ^{2}-2\epsilon'-\Delta)^{2}\label{eq:log-fail-prob-aqecc},
\end{equation}
where we used inequality~\eqref{eq:log-comb-number} on combinatorial numbers and defined $\alpha:=\frac{\delta-1}{n}$, $S(\alpha):=-\alpha \log\alpha - (1-\alpha)\log(1-\alpha)$. $\Delta=\Theta(d^{-n(1-\alpha)})$ decays exponentially. If we take $\epsilon'\sim\Delta$, and $\epsilon^2=\omega(\Delta)$, we see that Eq.~\eqref{eq:log-fail-prob-aqecc} becomes
\begin{equation}
    -\log\mathrm{Prob(fail)} \geq \Theta\left(-2Kn (1-\alpha)\log d + \frac{d^n\epsilon^4}{72\pi^3\log 2}\right).
\end{equation}
We can make the failure probability vanishingly small by setting $K\ll d^n\epsilon^4$. This implies that good quantum code with logical qudit number $k=\log_d K\propto n$, code distance $\delta\propto n$, with code proximity $\epsilon$ and generation failure probability vanishingly small in large $n$ limit is achievable from Haar random unitary.

\begin{proof}
To prove a subspace being an $\epsilon$-approximate QECC with high probability, we can see that the approximate pure QECC with distance $\delta$ is equivalent to the condition that any state in the code space is an $\epsilon$-approximate $\delta-1$-uniform state. This is equivalent to saying that $\max_{\ket{\psi}\in \mathcal{C}}\mathrm{Tr} (\ket{\psi}\bra{\psi}_{S}^{2}) < \frac{1}{d_{S}} + \epsilon ^{2}$ for any subsystem $S \subseteq [n]$ with size $|S|=\delta-1$. We first consider the probability for one fixed subsystem $S$ satisfying Eq.~\eqref{eq:altdef-aqecc}, with code space $\mathcal{C}$ being $K$-dimensional randomly-chosen subspace in total Hilbert space (denoted as $\mathcal C\in_R\mathcal G_K(\mathcal H)$):

\begin{equation}
\mathrm{Prob}_{\mathcal{C}\in _{R}\mathcal{G}_K(\mathcal {H})}\left(  \max _{\ket{\psi} \in \mathcal{C}} \mathrm{Tr}(\ket{\psi} \bra{\psi} _{S}^{2})< \frac{1}{d_{S}} + \epsilon ^{2} \right).
\end{equation}

For any $0<\epsilon'<1$, take an $\epsilon'$-net $\mathcal{N}_{\mathcal{C}}$ of subspace $\mathcal{C}$. For the maximal point $\ket{\psi}$, by definition, we can always find $\ket{\tilde{\psi}}\in \mathcal{N}_{\mathcal{C}}$ such that $\lvert \lvert \ket{\psi}-\ket{\tilde{\psi}}\rvert \rvert_{1}\leq\epsilon'$ and $\lvert \lvert  \ket{\psi}- \ket{\tilde{\psi}}\rvert \rvert_{2}\leq \frac{\epsilon'}{2}$. Using Lemma~\ref{lipschitz-pur}, Lipschitz constant of purity function is upper bounded by 4, we have
\begin{equation}
\left| \mathrm{Tr} (\ket{\tilde{\psi}} \bra{\tilde{\psi}} _{S}^{2})- \mathrm{Tr}(\ket{\psi} \bra{\psi} _{S}^{2})\right| \leq 4 \lvert \lvert  \ket{\tilde{\psi}} -\ket{\psi} \rvert \rvert _{2} \leq 2\epsilon',
\end{equation}
so the success probability on one specific subsystem $S$ is lower bounded by
\begin{equation}
\begin{aligned}
\mathrm{Prob}_{\mathcal{C}\in _{R}\mathcal{G}_K(\mathcal H)}\left(  \max _{\ket{\psi} \in \mathcal{C}} \mathrm{Tr}(\ket{\psi} \bra{\psi} _{S}^{2})  < \frac{1}{d_{S}} + \epsilon ^{2} \right)  &  \geq \mathrm{Prob}_{\mathcal{C}\in _{R}\mathcal{G}_K(\mathcal H)} \left( \max _{\ket{\tilde{\psi}} \in  \mathcal{N}_{\mathcal{C}}} \mathrm{Tr}(\ket{\tilde{\psi}} \bra{\tilde{\psi}} _{S}^{2} ) < \frac{1}{d_{S}} + \epsilon ^{2} - 2\epsilon' \right) \\
   & = 1- \mathrm{Prob}_{\mathcal{C}\in _{R}\mathcal{G}_K(\mathcal H)}\left( \max _{\ket{\tilde{\psi}}\in  \mathcal{N}_{\mathcal{C}} } \mathrm{Tr} (\ket{\tilde{\psi}} \bra{\tilde{\psi}} _{S}^{2}) \geq \frac{1}{d_{S}} + \epsilon ^{2}-2\epsilon'\right)\\
  & \geq 1-|\mathcal{N}_{\mathcal{C}}| \mathrm{Prob}_{{\ket{\psi} \sim \mu_{H}}} \left( \mathrm{Tr} (\ket{\psi} \bra{\psi} _{S}^{2}) \geq \frac{1}{d_{S}} + \epsilon ^{2} - 2\epsilon'\right)\\
  &\geq 1-|\mathcal{N}_{\mathcal{C}}| 2\exp \left(- \frac{d^{n}\mu ^{2}}{72\pi ^{3}\log2}\right),
\end{aligned}\label{eq:derivation-of-haar-approximate-qecc}
\end{equation}
where in the first inequality of Eq.~\eqref{eq:derivation-of-haar-approximate-qecc} we used the deduction that maximal purity on $\epsilon'$-net is less than $\frac 1{d_S}+\epsilon^2-2\epsilon'$ implies the proposition that the maximal value is less than $\frac{1}{d_S} + \epsilon^2$ on the total subspace; in the third inequality of Eq.~\eqref{eq:derivation-of-haar-approximate-qecc} we used the union theorem of probability, and from Lemma~\ref{lemma:epsilon-nets} we know that $|\mathcal{N}_{\mathcal{C}}|\leq (\frac{5}{\epsilon'})^{2K}$. The fourth inequality of Eq.~\eqref{eq:derivation-of-haar-approximate-qecc} comes from the previous Theorem~\ref{thm:approx-uniform-prob} on the concentration of $k$-uniform states with slight replacement ($\epsilon ^{2}\to\epsilon ^{2}-2\epsilon'$), where $\mu=\epsilon ^{2}-2\epsilon'-\Delta$, $\Delta= \frac{d^{\delta-1} + d^{n-(\delta-1)}}{d^{n}+1}- \frac{1}{d^{\delta-1}}$, and imposed the constraint that $\epsilon ^{2}-2\epsilon'\geq 2\Delta$. 

Using union theorem, the probability that all subsystem $|S|$ with $|S|=\delta-1$ is satisfied is thus lower bounded by

\begin{equation}
\begin{aligned}
    \mathrm{Prob}_{\mathcal{C}\in_{R}\mathcal{G}} \left(\forall S\subseteq[n],|S|=\delta-1, \max _{\ket{\psi} \in \mathcal{C}} \mathrm{Tr} (\ket{\psi} \bra{\psi} _{S}^{2}) < \frac{1}{d_{S}} + \epsilon ^{2} -2\epsilon'\right) &\geq 1- |\mathcal{N}_{\mathcal{C}}| {n\choose \delta-1} 2\exp \left(- \frac{d^{n}\mu ^{2}}{72\pi^{3}\log2}\right)\\
    &\geq  1-\left(\frac 5 {\epsilon'}\right)^{2K} {n\choose \delta-1} 2\exp \left(- \frac{d^{n}\mu ^{2}}{72\pi^{3}\log2}\right).
\end{aligned}
\end{equation}

This proves the probability bound of a random subspace being an $\epsilon$-approximate $((n,K,\delta))_{d}$ QECC.
    
\end{proof}

\subsection{Random Circuit Construction of Approximate QECCs}
We already see that Haar random unitary is capable of generating good approximate QECC with linear code distance, linear code rate and vanishingly small proximity and failure probability in Section~\ref{sec:Haar-random-aqecc}; however, in this section we show that this is not true for shallow depth random circuits. It turns out that to generate good approximate QECC with $k$ qudits we need circuit depth $L$ growing exponentially in $k$. To prove this, we use exactly the same techniques as in Section~\ref{sec:Haar-random-aqecc}. The only modification is to replace the usage of Theorem~\ref{thm:approx-uniform-prob} in Eq.~\eqref{eq:derivation-of-haar-approximate-qecc} with Theorem~\ref{thm:kdesign-kuniform}. The result is as follows:

\begin{theorem}\label{thm:random-circ-approximate-qecc}
    An $\epsilon$-approximate $((n,K,\delta))_d$ QECC can be generated from random 1 dimensional(1D) circuit ensemble of depth $L=O(nK^2\operatorname{poly}(k))$ with vanishingly small failure probability, where $k=\log_d K$ being the number of logical qudits, and given the constraint that $\epsilon=\omega(d^{-(n-k)/2})$ and $\epsilon=\omega({d^{-n/4}})$.
\end{theorem}

Note that the circuit depth have to grow exponentially with the qudit number $k:=\log_dK$. This illustrates a fundamental difference between Haar random unitary and random quantum circuit with shallow depth: shallow depth random circuit is not capable of generating "good" quantum codes with constant code rate $R:=\frac kn$ and linear distance $\delta=\Theta(n)$ with satisfying exponentially small proximity and failure probability under large $n$ limit.

\begin{proof}
Similar to the case of proving Haar random construction of approximate QECC in Section~\ref{sec:Haar-random-aqecc}, we see that approximate pure QECC with distance $\delta$ is equivalent to the condition that any state in the code space is an $\epsilon$-approximate $\delta-1$-uniform state. We first consider the probability for one fixed subsystem satisfying Eq.~\eqref{eq:altdef-aqecc}, with code space $\mathcal{C}$ generated by a quantum circuit ensemble forming approximate design from any fixed $K$-dimensional subspace:
\begin{equation}
\text{Prob}_{\mathcal{C}\in  \mu} \left( \max _{\ket{\psi} \in \mathcal{C}} \mathrm{Tr} (\ket{\psi} \bra{\psi} _{S}^{2}) < \frac{1}{d_{S}}+\epsilon ^{2}\right).
\end{equation}
Take an $\epsilon'$-net $\mathcal{N}_{\mathcal{C}}$ in the subspace $\mathcal{C}$. Using Lemma~\ref{lipschitz-pur}, we have
\begin{equation}
\lvert \mathrm{Tr} (\ket{\tilde{\psi}} \bra{\tilde{\psi}} _{S}^{2} ) - \mathrm{Tr} (\ket{\psi} \bra{\psi} _{S}^{2}) \rvert  \leq 4 \lvert \lvert  \ket{\tilde{\psi}} - \ket{\psi} \rvert \rvert _{2} \leq 2\epsilon',
\end{equation}
so the failure probability on subsystem $S$ can be upper bounded:
\begin{equation}
\begin{aligned}
\text{Prob}_{\mathcal{C}\in  \nu} \left( \max _{\ket{\psi} \in \mathcal{C}} \mathrm{Tr} (\ket{\psi} \bra{\psi} _{S}^{2}) < \frac{1}{d_{S}}+\epsilon ^{2}\right)  & \geq \text{Prob}_{\mathcal{C}\in \nu} \left( \max _{\ket{\tilde{\psi}} \in  \mathcal{N_{C}}} \mathrm{Tr} \left(\ket{\tilde{\psi}}\bra{\tilde{\psi}}_{S}^{2} \right)< \frac{1}{d_{S}} + \epsilon ^{2} - 2\epsilon' \right) \\
  & = 1- \text{Prob}_{\mathcal{C}\in \nu} \left( \max _{\ket{\tilde{\psi}} \in  \mathcal{N_{C}}} \mathrm{Tr} \left(\ket{\tilde{\psi}}\bra{\tilde{\psi}}_{S}^{2}\right) \geq \frac{1}{d_{S}} + \epsilon ^{2} - 2\epsilon' \right) \\
  & \geq 1- |\mathcal{N_{C}}| \text{Prob}_{\ket{\psi} \sim \nu} \left( \mathrm{Tr} (\ket{\psi }\bra{\psi}_{S}^{2})  \geq \frac{1}{d_{S}} + \epsilon ^{2} - 2\epsilon'\right) \\
 \text{Prob}_{\mathcal{C}\in  \nu} \left( \max _{\ket{\psi} \in \mathcal{C}} \mathrm{Tr} (\ket{\psi} \bra{\psi} _{S}^{2}) \geq \frac{1}{d_{S}}+\epsilon ^{2}\right)  & \leq |\mathcal{N_{C}}| \text{Prob}_{\ket{\psi} \sim \nu} \left( \mathrm{Tr} (\ket{\psi }\bra{\psi}_{S}^{2})  \geq \frac{1}{d_{S}} + \epsilon ^{2} - 2\epsilon'\right),\label{eq:derivartion-circ-aqecc}
\end{aligned}
\end{equation}
where we used $\nu$ to denote $\epsilon''$-approximate $t$ design ensemble, and $\mathcal C\in \nu$ means that $\mathcal C$ is generated by $U\sim\nu$ from a fixed subspace projector $P_0$ by $UP_0U^\dagger$. In the first inequality of Eq.~\eqref{eq:derivartion-circ-aqecc} we used the property of $\epsilon'$-net, in the second inequality of Eq.~\eqref{eq:derivartion-circ-aqecc} we used union theorem to convert the probability to one single state distributed by approximate design ensemble $\mu$, which has an upper bound according to Theorem~\ref{thm:kdesign-kuniform}. In the last inequality of Eq.~\eqref{eq:derivartion-circ-aqecc} we derived the upper bound of failure probability on one single subsystem $S$.

The total failure probability is thus lower bounded by
\begin{equation}
\begin{aligned}
\text{Prob}(\mathrm{fail})&:=\text{Prob}_{\mathcal{C}\in  \mu} \left( \exists S,|S| = \delta-1, \max _{\ket{\psi} \in \mathcal{C}} \mathrm{Tr} (\ket{\psi} \bra{\psi} _{S}^{2}) \geq \frac{1}{d_{S}}+\epsilon ^{2}\right)  \\ & \leq |\mathcal{N_{C}}| {n \choose \delta-1}\text{Prob}_{\ket{\psi} \sim \mu} \left( \mathrm{Tr} (\ket{\psi }\bra{\psi}_{S}^{2})  \geq \frac{1}{d_{S}} + \epsilon ^{2} - 2\epsilon'\right), \\
-\log \text{Prob}(\mathrm{fail})  & \gtrsim -2K \log \frac{5}{\epsilon'} - nS(\alpha) + t\log \sqrt{\epsilon ^2 - 2\epsilon'} + \frac n4 t\log d - \max\left( 0,- \log \frac 1{\epsilon''} + \left(\sigma + \frac 14\right) nt\log d \right) \\
&\simeq   -2K n (1-\alpha) \log d - nS(\alpha ) + t\log \epsilon + \frac n4 t\log d\\
 &=n\log d\left( -2K(1-\alpha) - S_d(\alpha) + \frac 14 t + \lambda t\right),\label{eq:derivation-of-circ-aqecc}
\end{aligned}
\end{equation}
for clarification of different $\epsilon$s, we are generating $\epsilon$-approximate codes using $\epsilon'$-net and $\epsilon''$-approximate $t$-design.
In the first inequality of Eq.~\eqref{eq:derivation-of-circ-aqecc} we used the union theorem over all possible partition $S$, in the second inequality we replaced the negative logarithm of each term with its lower bound: $-\log |\mathcal {N_C}| \geq 2K\log (\frac 5{\epsilon'})$, $-\log {n\choose \delta-1}\gtrsim -nS(\alpha)$ and Eq.~\eqref{derivation-approx-design-k-uniform} in Theorem~\ref{thm:kdesign-kuniform} with replacement $\epsilon\to \sqrt{\epsilon^2-2\epsilon'}$ (assuming $\epsilon^2>2\epsilon'$), $\alpha:=\frac{\delta-1}{n}$, and the given constraint $\epsilon = \omega(d^{-\frac{n(1-\alpha)}{2})})$. Then we simplify the outcome by assuming that $\log \frac{1}{\epsilon'} \geq (\sigma + \frac 14)nt\log d$, taking $\epsilon'=\Theta(d^{-n(1-\alpha)})$ and $\epsilon=\omega(d^{-(1-\alpha)n/2})$ so that $\epsilon^2-2\epsilon'\simeq \epsilon^2$, and removing all non-leading terms on the right hand side. Next we have the equality by taking $\lambda:=\frac 1n \log_d\epsilon$, $S_d(\alpha):=\frac{S(\alpha)}{\log d}$. Note that the failure probability is vanishingly small in large $n$ limit if we take $\lambda >-\frac 14$ (or $\epsilon=\omega(d^{-n/4})$ equivalently) and take $t=O(K)$ satisfying inequality $t> \frac{2K(1-\alpha) + S_d(\alpha)}{\lambda + \frac 14}$, resulting in circuit depth $L=O(nt^2\operatorname{polylog}(t))=O(nK^2\operatorname{poly}(k))$, where $k:=\log_dK$ is the logical qudit number in the generated approximate QECC.
\end{proof}

\subsection{Approximate quantum information masking}
In this section we investigate approximate QIM, and show the relation between approximate QECCs and approximate QIM. QIM is the process of distributing quantum information across a bipartite system in such a way that no individual subsystem contains sufficient information to reconstruct the original quantum information \cite{PhysRevLett.120.230501}.
QIM can be also generalized to $k$-uniform QIM with any $k$ subsystems can not recover the original quantum information  \cite{PhysRevA.104.032601}: an  operation $\mathcal S$ is said to \emph{$k$-uniformly mask}  quantum information contained in states $\{\ket{\varphi_i}\mid \ket{\varphi_i}\in \bbC^K\}$ to states $\{\ket{\psi_i}=\mathcal S(\ket{\varphi_i})\in (\bbC^d)^{\otimes n}\}$ if, for all states $\ket{\psi_i}$   and all subsystems $S\subseteq [n]$ with $|S|=k$, $\ketbra{\psi_i}{\psi_i}_S$ are identical. Clearly, if a pure  $((n,K,\delta))_d$ QECC exists, then it is possible to $(\delta-1)$ uniformly mask all  states of $\mathbb C^K$ to $(\bbC^d)^{\otimes n}$.
However, 
QIM tasks under specific circumstances are proved not possible \cite{PhysRevLett.120.230501,PhysRevA.104.032601}, and it is not possible to make measurements to local subsystem with infinite precision due to unavoidable noise. So it is enough to relax QIM restriction to its approximate version. 
\begin{definition}[$\epsilon$-approximate $k$-uniform QIM]
An  operation $\mathcal S$ is said to \emph{ $\epsilon$-approximate $k$-uniformly mask}  quantum information contained in states $\{\ket{\varphi_i}\mid \ket{\varphi_i}\in \bbC^K\}$ to states $\{\ket{\psi_i}=\mathcal S(\ket{\varphi_i})\in (\bbC^d)^{\otimes n}\}$ if, for all states $\ket{\psi_i},\ket{\psi_j}$   and all subsystems $S\subseteq [n]$ with $|S|=k$, 
\begin{equation}
 || \ket{\psi_i}\bra{\psi_i}_S - \ket{\psi_j}\bra{\psi_j}_S ||\leq \epsilon .
\end{equation}
\end{definition}

Similar to the relationship between QECCs and QIM, we can establish a connection between approximate QECCs and approximate QIM.


\begin{theorem}
    If an $\epsilon$-approximate pure $((n,K,\delta))_{d}$  QECC exists, then it is possible to $2\epsilon$-approximate $(\delta-1)$-uniformly mask all states of $\bbC^K$ to $(\bbC^d)^{\otimes n}$.
\end{theorem}
\begin{proof}
    According to Theorem~\ref{thm:aqecc-eq-uniform}, each state of the $\epsilon$-approximate pure $((n,K,\delta))_{d}$  QECC $\cC$ is an $\epsilon$-approximate $(\delta-1)$-uniform state. Let $\mathcal S$ be the isomorphism between $\bbC^K$ and $\cC$,  then for all states $\ket{\psi_i},\ket{\psi_j}\in\cC$ and all subsystems $S\subseteq [n]$ with $|S|=\delta-1$,
\begin{equation}
    \begin{aligned}
|| \ket{\psi_i}\bra{\psi_i}_S - \ket{\psi_j}\bra{\psi_j}_S ||&= \left|\left| \frac{I_S+(P_{i})_S}{d_S} -\frac{I_S+(P_{j})_S}{d_S}  \right|\right|\\ 
&= \frac{1}{d_S} \left|\left|(P_i)_S-(P_j)_S\right|\right|\\
&\leq \frac{1}{d_S} \left(\left|\left|(P_i)_S\right|\right|+ \left|\left|(P_j)_S\right|\right|\right)\\
&\leq 2\epsilon,        
    \end{aligned}
\end{equation}
where we have used alternative Definition~\ref{alternate-def2} of $\epsilon$-approximate $k$-uniform state and triangle inequality. According to the definition of $\epsilon$-approximate QIM, the operation  $2\epsilon$-approximate $(\delta-1)$-uniformly mask all states of $\bbC^K$ to $(\bbC^d)^{\otimes n}$.
\end{proof}

This theorem shows that using the process of encoding to approximate QECC, we can distribute logical quantum information to a larger physical system where the information is approximately and globally hided, so that almost no information can be gained by small number of parties. This theorem also shows the strong connection between approximate QIM and approximate QECC.

\section{Numerical simulations}\label{num}

To study the existence of approximate $k$-uniform states, we run a simple numerical simulation to find optimal $\epsilon$ of approximate uniform state under small system number and local dimension. The simulation is designed as the following optimization problem:
\begin{equation}
    \begin{aligned}
        \epsilon^* = \underset{\substack{c_{j} \\ j=0,1,\cdots,d^n-1}}{\mathrm{minimize }} & \max_{A\subseteq [n],|A|=k} \sqrt{\tr(\rho_A^2)- \frac {1}{d^k}},\\
   \text{subject to } &\begin{cases}\rho_A=\mathrm{Tr}_{A^c}  \ket\psi\bra\psi\\
    \ket\psi= \sum_{j=0}^{d^n-1}c_j\ket {j_1j_2\cdots j_n}\\
    \sum_{j=0}^{d^n-1} |c_{j}|^2=1\end{cases}   ,
    \end{aligned}
\end{equation}
where the parameters are $c_j\in \mathbb C\ (j=0,\cdots, d^n-1) $, and $j=(j_1j_2\cdots j_n)$ is the representation of integer $j$ in base $d$. This forms an optimization problem with $d^{n}$ complex parameters, or $2d^{n}$ real parameters equivalently and makes the 
 problem of finding approximate $k$-uniform state optimizable for small $d,n$ cases. Here is a table illustrating our simulation results of approximate AME and uniform states. The number in each box represents an optimized approximate $k$-uniform or AME state with specific $\epsilon$. For example, for $n=4,d=2$, it means that there exists an approximate $\mathrm{AME}(4,2)$ state with $\epsilon\approx 0.2887$, and approximate $1$-uniform state with $\epsilon\approx 0.0077$. We can see that for AME$(4,2)$, $0.2887> \frac 1{2\sqrt {19}}\approx 0.1147$ satisfies our computed bounds. Moreover, the notion of approximate uniform states can be generalized to the case where $k>\fl{n}{2}$, which is not possible for the exact case. For example, numerical results show that there exists $0.3536$-approximate $3$-uniform states on 5 qubits (note that this is close to the exact 2-uniform case, with $\epsilon=\sqrt{\frac{1}{4}-\frac{1}{8}}\approx 0.3536 $). This could be useful for tasks like approximate QIM. As the $n,d$ increases, the number of dimension increases as $n^d$, which takes exponentially long time to run on the computer.

\begin{table}[h]
\caption{Numerical results on approximate $k$-uniform and AME$(n,d)$ states}
\centering
\begin{tabular}{cc|c|c|c|c|c|c}
\toprule
 &&\multicolumn{3}{c}{$k=\lfloor n/2\rfloor$}\vline&\multicolumn{3}{c}{$k=\lfloor n/2\rfloor-1$} \\\hline
\multicolumn{2}{c}{\multirow{2}*{$\epsilon$}}\vline&\multicolumn{3}{c}{$d$} \vline& \multicolumn{3}{c}{$d$}   \\
\cline{3-8} 
 && 2 & 3 & 4 & 2&3&4 \\ 
\hline
\multicolumn{1}{c|}{\multirow{5}{*}{$n$}} &4& \lpurple $0.2887$ & \lgreen $2.8170\times 10^{-7}$&  \lgreen$0.1781$ &\lgreen $0.0077$ &$1.600\times 10^{-6}$ \lgreen & $3.726\times 10^{-6}$ \lgreen\\
\multicolumn{1}{c|}{}&5& \lgreen $3.501\times 10^{-7}$   &\lgreen 0.0387& 0.1322\lgreen & $3.371\times 10^{-5}$\lgreen & 0.0004\lgreen &\lgreen \\   
 \multicolumn{1}{c|}{}&6&  \lgreen 0.2473&  \lgreen& \lgreen &0.0714\lgreen & \lgreen & \lgreen\\   
 \multicolumn{1}{c|}{}  &7& \lpurple $0.1214$& \lgreen & \lgreen & 0.0138\lgreen & \lgreen & \lgreen\\  
\multicolumn{1}{c|}{}&8& \lpurple&\lpurple &  &0.1228 \lgreen & \lgreen & \lgreen  \\
\end{tabular}

The numerical results on approximate $k$-uniform and AME states. Grids are colored light purple if $k$-uniform states are non-existent, green if $k$-uniform states exist, and white if the existence of $k$-uniform states are unknown.
\label{tab:d_n_table}
\end{table}

\section{Summary and future directions}
In this work, we defined the concept of approximate $k$-uniform states, proved that existence of approximate $k$-uniform states is similar to the existence of exact ones, and showed that random (either Haar random or shallow depth) circuits give rise to such states with very high probability under large $n$ limit. Next we investigated its application to approximate QECC, relating code proximity to bounds on quantum weight enumerators, and showed the performance of random unitary and random circuit approach for constructing approximate QECC. Then we provided potential application of approximate $k$-uniform states in approximate QIM. 

Our results have many benefits. The exact $k$-uniform states are mostly constructed through combinatorial methods like orthogonal arrays and Latin squares, and AME$(n,d)$ states for some cases are proved to be impossible. Our generalization of $k$-uniform states to approximate cases makes some approximate AME states possible and $k$-uniform states more accessible. Our result also provides practical usage to the real-world quantum information processing tasks. Inspired by approximate $k$-uniform states, we defined approximate QECC, which relates the notion of quantum weight enumerators. The approximate QECCs are shown to be efficiently constructed from Haar random unitaries but not efficiently constructed by shallow random quantum circuits. This shows a gap between ideal random unitaries and its approximation implemented by random circuit ensemble, which leaves an interesting open question for further study.

\section{Acknowledgments}
After this work was completed, we became aware of related work by Xiaodi Li, Xinyang Shu and Huangjun Zhu that also analyzes approximate QIM using Haar random tools. We thank them for letting us know about their work.

We thank Xiande Zhang, Zhiyao Lu, Yue Cao, Wenjun Yu, Xingjian Zhang, Tianfeng Feng, Jue Xu, Jiayi Wu and Yue Wang for their helpful discussion and suggestion.

This work acknowledges funding from Innovation Program for Quantum Science and Technology via Project 2024ZD0301900, National Natural Science Foundation of China (NSFC) via Project No. 12347104 and No. 12305030, Guangdong Basic and Applied Basic Research Foundation via Project 2023A1515012185, Hong Kong Research Grant Council (RGC) via No. 27300823, N\_HKU718/23, and R6010-23, Guangdong Provincial Quantum Science Strategic Initiative No. GDZX2303007, HKU Seed Fund for Basic Research for New Staff via Project 2201100596. Y.Z. is also supported by the Innovation Program for Quantum Science and Technology Grant Nos.~2021ZD0302000, the National Natural Science Foundation of China (NSFC) Grant No.~12205048, the Shanghai Science and Technology Innovation Action Plan Grant No.~24LZ1400200, the Shanghai Pilot Program for Basic Research - Fudan University 21TQ1400100 (25TQ003), and the start-up funding of Fudan University.


\nocite{*}
\bibliographystyle{apsrev4-2}
\bibliography{ref}

\begin{thebibliography}{58}%
\makeatletter
\providecommand \@ifxundefined [1]{%
 \@ifx{#1\undefined}
}%
\providecommand \@ifnum [1]{%
 \ifnum #1\expandafter \@firstoftwo
 \else \expandafter \@secondoftwo
 \fi
}%
\providecommand \@ifx [1]{%
 \ifx #1\expandafter \@firstoftwo
 \else \expandafter \@secondoftwo
 \fi
}%
\providecommand \natexlab [1]{#1}%
\providecommand \enquote  [1]{``#1''}%
\providecommand \bibnamefont  [1]{#1}%
\providecommand \bibfnamefont [1]{#1}%
\providecommand \citenamefont [1]{#1}%
\providecommand \href@noop [0]{\@secondoftwo}%
\providecommand \href [0]{\begingroup \@sanitize@url \@href}%
\providecommand \@href[1]{\@@startlink{#1}\@@href}%
\providecommand \@@href[1]{\endgroup#1\@@endlink}%
\providecommand \@sanitize@url [0]{\catcode `\\12\catcode `\$12\catcode `\&12\catcode `\#12\catcode `\^12\catcode `\_12\catcode `\%12\relax}%
\providecommand \@@startlink[1]{}%
\providecommand \@@endlink[0]{}%
\providecommand \url  [0]{\begingroup\@sanitize@url \@url }%
\providecommand \@url [1]{\endgroup\@href {#1}{\urlprefix }}%
\providecommand \urlprefix  [0]{URL }%
\providecommand \Eprint [0]{\href }%
\providecommand \doibase [0]{https://doi.org/}%
\providecommand \selectlanguage [0]{\@gobble}%
\providecommand \bibinfo  [0]{\@secondoftwo}%
\providecommand \bibfield  [0]{\@secondoftwo}%
\providecommand \translation [1]{[#1]}%
\providecommand \BibitemOpen [0]{}%
\providecommand \bibitemStop [0]{}%
\providecommand \bibitemNoStop [0]{.\EOS\space}%
\providecommand \EOS [0]{\spacefactor3000\relax}%
\providecommand \BibitemShut  [1]{\csname bibitem#1\endcsname}%
\let\auto@bib@innerbib\@empty
\bibitem [{\citenamefont {Scott}(2004)}]{PhysRevA.69.052330}%
  \BibitemOpen
  \bibfield  {author} {\bibinfo {author} {\bibfnamefont {A.~J.}\ \bibnamefont {Scott}},\ }\href {https://doi.org/10.1103/PhysRevA.69.052330} {\bibfield  {journal} {\bibinfo  {journal} {Phys. Rev. A}\ }\textbf {\bibinfo {volume} {69}},\ \bibinfo {pages} {052330} (\bibinfo {year} {2004})}\BibitemShut {NoStop}%
\bibitem [{\citenamefont {Helwig}\ \emph {et~al.}(2012)\citenamefont {Helwig}, \citenamefont {Cui}, \citenamefont {Latorre}, \citenamefont {Riera},\ and\ \citenamefont {Lo}}]{PhysRevA.86.052335}%
  \BibitemOpen
  \bibfield  {author} {\bibinfo {author} {\bibfnamefont {W.}~\bibnamefont {Helwig}}, \bibinfo {author} {\bibfnamefont {W.}~\bibnamefont {Cui}}, \bibinfo {author} {\bibfnamefont {J.~I.}\ \bibnamefont {Latorre}}, \bibinfo {author} {\bibfnamefont {A.}~\bibnamefont {Riera}},\ and\ \bibinfo {author} {\bibfnamefont {H.-K.}\ \bibnamefont {Lo}},\ }\href {https://doi.org/10.1103/PhysRevA.86.052335} {\bibfield  {journal} {\bibinfo  {journal} {Phys. Rev. A}\ }\textbf {\bibinfo {volume} {86}},\ \bibinfo {pages} {052335} (\bibinfo {year} {2012})}\BibitemShut {NoStop}%
\bibitem [{\citenamefont {Cleve}\ \emph {et~al.}(1999)\citenamefont {Cleve}, \citenamefont {Gottesman},\ and\ \citenamefont {Lo}}]{PhysRevLett.83.648}%
  \BibitemOpen
  \bibfield  {author} {\bibinfo {author} {\bibfnamefont {R.}~\bibnamefont {Cleve}}, \bibinfo {author} {\bibfnamefont {D.}~\bibnamefont {Gottesman}},\ and\ \bibinfo {author} {\bibfnamefont {H.-K.}\ \bibnamefont {Lo}},\ }\href {https://doi.org/10.1103/PhysRevLett.83.648} {\bibfield  {journal} {\bibinfo  {journal} {Phys. Rev. Lett.}\ }\textbf {\bibinfo {volume} {83}},\ \bibinfo {pages} {648} (\bibinfo {year} {1999})}\BibitemShut {NoStop}%
\bibitem [{\citenamefont {Shi}\ \emph {et~al.}(2021)\citenamefont {Shi}, \citenamefont {Li}, \citenamefont {Chen},\ and\ \citenamefont {Zhang}}]{PhysRevA.104.032601}%
  \BibitemOpen
  \bibfield  {author} {\bibinfo {author} {\bibfnamefont {F.}~\bibnamefont {Shi}}, \bibinfo {author} {\bibfnamefont {M.-S.}\ \bibnamefont {Li}}, \bibinfo {author} {\bibfnamefont {L.}~\bibnamefont {Chen}},\ and\ \bibinfo {author} {\bibfnamefont {X.}~\bibnamefont {Zhang}},\ }\href {https://doi.org/10.1103/PhysRevA.104.032601} {\bibfield  {journal} {\bibinfo  {journal} {Phys. Rev. A}\ }\textbf {\bibinfo {volume} {104}},\ \bibinfo {pages} {032601} (\bibinfo {year} {2021})}\BibitemShut {NoStop}%
\bibitem [{\citenamefont {Zhao}\ \emph {et~al.}(2024)\citenamefont {Zhao}, \citenamefont {Zhou},\ and\ \citenamefont {Childs}}]{zhao2024entanglement}%
  \BibitemOpen
  \bibfield  {author} {\bibinfo {author} {\bibfnamefont {Q.}~\bibnamefont {Zhao}}, \bibinfo {author} {\bibfnamefont {Y.}~\bibnamefont {Zhou}},\ and\ \bibinfo {author} {\bibfnamefont {A.~M.}\ \bibnamefont {Childs}},\ }\href {https://arxiv.org/abs/2406.02379} {\bibfield  {journal} {\bibinfo  {journal} {arXiv:2406.02379}\ } (\bibinfo {year} {2024})}\BibitemShut {NoStop}%
\bibitem [{\citenamefont {Horodecki}\ \emph {et~al.}(2022)\citenamefont {Horodecki}, \citenamefont {Rudnicki},\ and\ \citenamefont {\ifmmode~\dot{Z}\else \.{Z}\fi{}yczkowski}}]{PRXQuantum.3.010101}%
  \BibitemOpen
  \bibfield  {author} {\bibinfo {author} {\bibfnamefont {P.}~\bibnamefont {Horodecki}}, \bibinfo {author} {\bibfnamefont {L.}~\bibnamefont {Rudnicki}},\ and\ \bibinfo {author} {\bibfnamefont {K.}~\bibnamefont {\ifmmode~\dot{Z}\else \.{Z}\fi{}yczkowski}},\ }\href {https://doi.org/10.1103/PRXQuantum.3.010101} {\bibfield  {journal} {\bibinfo  {journal} {PRX Quantum}\ }\textbf {\bibinfo {volume} {3}},\ \bibinfo {pages} {010101} (\bibinfo {year} {2022})}\BibitemShut {NoStop}%
\bibitem [{\citenamefont {Rather}\ \emph {et~al.}(2022)\citenamefont {Rather}, \citenamefont {Burchardt}, \citenamefont {Bruzda}, \citenamefont {Rajchel-Mieldzio\ifmmode~\acute{c}\else \'{c}\fi{}}, \citenamefont {Lakshminarayan},\ and\ \citenamefont {\ifmmode~\dot{Z}\else \.{Z}\fi{}yczkowski}}]{PhysRevLett.128.080507}%
  \BibitemOpen
  \bibfield  {author} {\bibinfo {author} {\bibfnamefont {S.~A.}\ \bibnamefont {Rather}}, \bibinfo {author} {\bibfnamefont {A.}~\bibnamefont {Burchardt}}, \bibinfo {author} {\bibfnamefont {W.}~\bibnamefont {Bruzda}}, \bibinfo {author} {\bibfnamefont {G.}~\bibnamefont {Rajchel-Mieldzio\ifmmode~\acute{c}\else \'{c}\fi{}}}, \bibinfo {author} {\bibfnamefont {A.}~\bibnamefont {Lakshminarayan}},\ and\ \bibinfo {author} {\bibfnamefont {K.}~\bibnamefont {\ifmmode~\dot{Z}\else \.{Z}\fi{}yczkowski}},\ }\href {https://doi.org/10.1103/PhysRevLett.128.080507} {\bibfield  {journal} {\bibinfo  {journal} {Phys. Rev. Lett.}\ }\textbf {\bibinfo {volume} {128}},\ \bibinfo {pages} {080507} (\bibinfo {year} {2022})}\BibitemShut {NoStop}%
\bibitem [{\citenamefont {Feng}\ \emph {et~al.}(2017)\citenamefont {Feng}, \citenamefont {Jin}, \citenamefont {Xing},\ and\ \citenamefont {Yuan}}]{feng2017multipartite}%
  \BibitemOpen
  \bibfield  {author} {\bibinfo {author} {\bibfnamefont {K.}~\bibnamefont {Feng}}, \bibinfo {author} {\bibfnamefont {L.}~\bibnamefont {Jin}}, \bibinfo {author} {\bibfnamefont {C.}~\bibnamefont {Xing}},\ and\ \bibinfo {author} {\bibfnamefont {C.}~\bibnamefont {Yuan}},\ }\href {https://doi.org/10.1109/TIT.2017.2700866} {\bibfield  {journal} {\bibinfo  {journal} {IEEE Trans. Inf. Theory}\ }\textbf {\bibinfo {volume} {63}},\ \bibinfo {pages} {5618} (\bibinfo {year} {2017})}\BibitemShut {NoStop}%
\bibitem [{\citenamefont {Goyeneche}\ and\ \citenamefont {\ifmmode~\dot{Z}\else \.{Z}\fi{}yczkowski}(2014)}]{goyeneche2014genuinely}%
  \BibitemOpen
  \bibfield  {author} {\bibinfo {author} {\bibfnamefont {D.}~\bibnamefont {Goyeneche}}\ and\ \bibinfo {author} {\bibfnamefont {K.}~\bibnamefont {\ifmmode~\dot{Z}\else \.{Z}\fi{}yczkowski}},\ }\href {https://doi.org/10.1103/PhysRevA.90.022316} {\bibfield  {journal} {\bibinfo  {journal} {Phys. Rev. A}\ }\textbf {\bibinfo {volume} {90}},\ \bibinfo {pages} {022316} (\bibinfo {year} {2014})}\BibitemShut {NoStop}%
\bibitem [{\citenamefont {Li}\ and\ \citenamefont {Wang}(2019{\natexlab{a}})}]{li2019k}%
  \BibitemOpen
  \bibfield  {author} {\bibinfo {author} {\bibfnamefont {M.-S.}\ \bibnamefont {Li}}\ and\ \bibinfo {author} {\bibfnamefont {Y.-L.}\ \bibnamefont {Wang}},\ }\href {https://doi.org/10.1103/PhysRevA.99.042332} {\bibfield  {journal} {\bibinfo  {journal} {Phys. Rev. A}\ }\textbf {\bibinfo {volume} {99}},\ \bibinfo {pages} {042332} (\bibinfo {year} {2019}{\natexlab{a}})}\BibitemShut {NoStop}%
\bibitem [{\citenamefont {Pang}\ \emph {et~al.}(2019)\citenamefont {Pang}, \citenamefont {Zhang}, \citenamefont {Lin},\ and\ \citenamefont {Zhang}}]{pang2019two}%
  \BibitemOpen
  \bibfield  {author} {\bibinfo {author} {\bibfnamefont {S.-Q.}\ \bibnamefont {Pang}}, \bibinfo {author} {\bibfnamefont {X.}~\bibnamefont {Zhang}}, \bibinfo {author} {\bibfnamefont {X.}~\bibnamefont {Lin}},\ and\ \bibinfo {author} {\bibfnamefont {Q.-J.}\ \bibnamefont {Zhang}},\ }\href {https://www.nature.com/articles/s41534-019-0165-8#citeas} {\bibfield  {journal} {\bibinfo  {journal} {npj Quantum Infor.}\ }\textbf {\bibinfo {volume} {5}},\ \bibinfo {pages} {52} (\bibinfo {year} {2019})}\BibitemShut {NoStop}%
\bibitem [{\citenamefont {Goyeneche}\ \emph {et~al.}(2015)\citenamefont {Goyeneche}, \citenamefont {Alsina}, \citenamefont {Latorre}, \citenamefont {Riera},\ and\ \citenamefont {\ifmmode~\dot{Z}\else \.{Z}\fi{}yczkowski}}]{goyeneche2015absolutely}%
  \BibitemOpen
  \bibfield  {author} {\bibinfo {author} {\bibfnamefont {D.}~\bibnamefont {Goyeneche}}, \bibinfo {author} {\bibfnamefont {D.}~\bibnamefont {Alsina}}, \bibinfo {author} {\bibfnamefont {J.~I.}\ \bibnamefont {Latorre}}, \bibinfo {author} {\bibfnamefont {A.}~\bibnamefont {Riera}},\ and\ \bibinfo {author} {\bibfnamefont {K.}~\bibnamefont {\ifmmode~\dot{Z}\else \.{Z}\fi{}yczkowski}},\ }\href {https://doi.org/10.1103/PhysRevA.92.032316} {\bibfield  {journal} {\bibinfo  {journal} {Phys. Rev. A}\ }\textbf {\bibinfo {volume} {92}},\ \bibinfo {pages} {032316} (\bibinfo {year} {2015})}\BibitemShut {NoStop}%
\bibitem [{\citenamefont {Goyeneche}\ \emph {et~al.}(2018)\citenamefont {Goyeneche}, \citenamefont {Raissi}, \citenamefont {Di~Martino},\ and\ \citenamefont {\ifmmode~\dot{Z}\else \.{Z}\fi{}yczkowski}}]{PhysRevA.97.062326}%
  \BibitemOpen
  \bibfield  {author} {\bibinfo {author} {\bibfnamefont {D.}~\bibnamefont {Goyeneche}}, \bibinfo {author} {\bibfnamefont {Z.}~\bibnamefont {Raissi}}, \bibinfo {author} {\bibfnamefont {S.}~\bibnamefont {Di~Martino}},\ and\ \bibinfo {author} {\bibfnamefont {K.}~\bibnamefont {\ifmmode~\dot{Z}\else \.{Z}\fi{}yczkowski}},\ }\href {https://doi.org/10.1103/PhysRevA.97.062326} {\bibfield  {journal} {\bibinfo  {journal} {Phys. Rev. A}\ }\textbf {\bibinfo {volume} {97}},\ \bibinfo {pages} {062326} (\bibinfo {year} {2018})}\BibitemShut {NoStop}%
\bibitem [{\citenamefont {Raissi}\ \emph {et~al.}(2020)\citenamefont {Raissi}, \citenamefont {Teixid\'o}, \citenamefont {Gogolin},\ and\ \citenamefont {Ac\'{\i}n}}]{raissi2020constructions}%
  \BibitemOpen
  \bibfield  {author} {\bibinfo {author} {\bibfnamefont {Z.}~\bibnamefont {Raissi}}, \bibinfo {author} {\bibfnamefont {A.}~\bibnamefont {Teixid\'o}}, \bibinfo {author} {\bibfnamefont {C.}~\bibnamefont {Gogolin}},\ and\ \bibinfo {author} {\bibfnamefont {A.}~\bibnamefont {Ac\'{\i}n}},\ }\href {https://doi.org/10.1103/PhysRevResearch.2.033411} {\bibfield  {journal} {\bibinfo  {journal} {Phys. Rev. Res.}\ }\textbf {\bibinfo {volume} {2}},\ \bibinfo {pages} {033411} (\bibinfo {year} {2020})}\BibitemShut {NoStop}%
\bibitem [{\citenamefont {Zang}\ \emph {et~al.}(2021)\citenamefont {Zang}, \citenamefont {Facchi},\ and\ \citenamefont {Tian}}]{Zang_2021}%
  \BibitemOpen
  \bibfield  {author} {\bibinfo {author} {\bibfnamefont {Y.}~\bibnamefont {Zang}}, \bibinfo {author} {\bibfnamefont {P.}~\bibnamefont {Facchi}},\ and\ \bibinfo {author} {\bibfnamefont {Z.}~\bibnamefont {Tian}},\ }\href {https://doi.org/10.1088/1751-8121/ac3705} {\bibfield  {journal} {\bibinfo  {journal} {J. Phys. A-Math. Theor.}\ }\textbf {\bibinfo {volume} {54}},\ \bibinfo {pages} {505204} (\bibinfo {year} {2021})}\BibitemShut {NoStop}%
\bibitem [{\citenamefont {Shi}\ \emph {et~al.}(2022)\citenamefont {Shi}, \citenamefont {Shen}, \citenamefont {Chen},\ and\ \citenamefont {Zhang}}]{Heterogeneous_Systems}%
  \BibitemOpen
  \bibfield  {author} {\bibinfo {author} {\bibfnamefont {F.}~\bibnamefont {Shi}}, \bibinfo {author} {\bibfnamefont {Y.}~\bibnamefont {Shen}}, \bibinfo {author} {\bibfnamefont {L.}~\bibnamefont {Chen}},\ and\ \bibinfo {author} {\bibfnamefont {X.}~\bibnamefont {Zhang}},\ }\href {https://doi.org/https://ieeexplore.ieee.org/document/9698018/} {\bibfield  {journal} {\bibinfo  {journal} {IEEE Trans. Inf. Theory}\ }\textbf {\bibinfo {volume} {68}},\ \bibinfo {pages} {3115} (\bibinfo {year} {2022})}\BibitemShut {NoStop}%
\bibitem [{\citenamefont {Feng}\ \emph {et~al.}(2023)\citenamefont {Feng}, \citenamefont {Jin}, \citenamefont {Xing},\ and\ \citenamefont {Yuan}}]{10143323}%
  \BibitemOpen
  \bibfield  {author} {\bibinfo {author} {\bibfnamefont {K.}~\bibnamefont {Feng}}, \bibinfo {author} {\bibfnamefont {L.}~\bibnamefont {Jin}}, \bibinfo {author} {\bibfnamefont {C.}~\bibnamefont {Xing}},\ and\ \bibinfo {author} {\bibfnamefont {C.}~\bibnamefont {Yuan}},\ }\href {https://doi.org/10.1109/TIT.2023.3282221} {\bibfield  {journal} {\bibinfo  {journal} {IEEE Trans. Inf. Theory}\ }\textbf {\bibinfo {volume} {69}},\ \bibinfo {pages} {5845} (\bibinfo {year} {2023})}\BibitemShut {NoStop}%
\bibitem [{\citenamefont {Rains}(1999)}]{796376}%
  \BibitemOpen
  \bibfield  {author} {\bibinfo {author} {\bibfnamefont {E.}~\bibnamefont {Rains}},\ }\href {https://doi.org/10.1109/18.796376} {\bibfield  {journal} {\bibinfo  {journal} {IEEE Trans. Inf. Theory}\ }\textbf {\bibinfo {volume} {45}},\ \bibinfo {pages} {2361} (\bibinfo {year} {1999})}\BibitemShut {NoStop}%
\bibitem [{\citenamefont {Higuchi}\ and\ \citenamefont {Sudbery}(2000)}]{higuchi2000entangled}%
  \BibitemOpen
  \bibfield  {author} {\bibinfo {author} {\bibfnamefont {A.}~\bibnamefont {Higuchi}}\ and\ \bibinfo {author} {\bibfnamefont {A.}~\bibnamefont {Sudbery}},\ }\href {https://doi.org/https://doi.org/10.1016/S0375-9601(00)00480-1} {\bibfield  {journal} {\bibinfo  {journal} {Phys. Lett. A}\ }\textbf {\bibinfo {volume} {273}},\ \bibinfo {pages} {213} (\bibinfo {year} {2000})}\BibitemShut {NoStop}%
\bibitem [{\citenamefont {Huber}\ \emph {et~al.}(2017)\citenamefont {Huber}, \citenamefont {G\"uhne},\ and\ \citenamefont {Siewert}}]{PhysRevLett.118.200502}%
  \BibitemOpen
  \bibfield  {author} {\bibinfo {author} {\bibfnamefont {F.}~\bibnamefont {Huber}}, \bibinfo {author} {\bibfnamefont {O.}~\bibnamefont {G\"uhne}},\ and\ \bibinfo {author} {\bibfnamefont {J.}~\bibnamefont {Siewert}},\ }\href {https://doi.org/10.1103/PhysRevLett.118.200502} {\bibfield  {journal} {\bibinfo  {journal} {Phys. Rev. Lett.}\ }\textbf {\bibinfo {volume} {118}},\ \bibinfo {pages} {200502} (\bibinfo {year} {2017})}\BibitemShut {NoStop}%
\bibitem [{\citenamefont {Huber}\ \emph {et~al.}(2018)\citenamefont {Huber}, \citenamefont {Eltschka}, \citenamefont {Siewert},\ and\ \citenamefont {G{\"u}hne}}]{huber2018bounds}%
  \BibitemOpen
  \bibfield  {author} {\bibinfo {author} {\bibfnamefont {F.}~\bibnamefont {Huber}}, \bibinfo {author} {\bibfnamefont {C.}~\bibnamefont {Eltschka}}, \bibinfo {author} {\bibfnamefont {J.}~\bibnamefont {Siewert}},\ and\ \bibinfo {author} {\bibfnamefont {O.}~\bibnamefont {G{\"u}hne}},\ }\href {https://iopscience.iop.org/article/10.1088/1751-8121/aaade5/meta} {\bibfield  {journal} {\bibinfo  {journal} {J. Phys. A-Math. Theor.}\ }\textbf {\bibinfo {volume} {51}},\ \bibinfo {pages} {175301} (\bibinfo {year} {2018})}\BibitemShut {NoStop}%
\bibitem [{\citenamefont {Shi}\ \emph {et~al.}(2025)\citenamefont {Shi}, \citenamefont {Ning}, \citenamefont {Zhao},\ and\ \citenamefont {Zhang}}]{10718358}%
  \BibitemOpen
  \bibfield  {author} {\bibinfo {author} {\bibfnamefont {F.}~\bibnamefont {Shi}}, \bibinfo {author} {\bibfnamefont {Y.}~\bibnamefont {Ning}}, \bibinfo {author} {\bibfnamefont {Q.}~\bibnamefont {Zhao}},\ and\ \bibinfo {author} {\bibfnamefont {X.}~\bibnamefont {Zhang}},\ }\href {https://doi.org/10.1109/TIT.2024.3481042} {\bibfield  {journal} {\bibinfo  {journal} {IEEE Trans. Inf. Theory}\ }\textbf {\bibinfo {volume} {71}},\ \bibinfo {pages} {413} (\bibinfo {year} {2025})}\BibitemShut {NoStop}%
\bibitem [{\citenamefont {Ning}\ \emph {et~al.}(2025)\citenamefont {Ning}, \citenamefont {Shi}, \citenamefont {Luo},\ and\ \citenamefont {Zhang}}]{ning2025linear}%
  \BibitemOpen
  \bibfield  {author} {\bibinfo {author} {\bibfnamefont {Y.}~\bibnamefont {Ning}}, \bibinfo {author} {\bibfnamefont {F.}~\bibnamefont {Shi}}, \bibinfo {author} {\bibfnamefont {T.}~\bibnamefont {Luo}},\ and\ \bibinfo {author} {\bibfnamefont {X.}~\bibnamefont {Zhang}},\ }\href {https://arxiv.org/abs/2503.02222} {\bibfield  {journal} {\bibinfo  {journal} {arXiv:2503.02222}\ } (\bibinfo {year} {2025})}\BibitemShut {NoStop}%
\bibitem [{\citenamefont {Huber}\ and\ \citenamefont {Wyderka}()}]{AMEtable}%
  \BibitemOpen
  \bibfield  {author} {\bibinfo {author} {\bibfnamefont {F.}~\bibnamefont {Huber}}\ and\ \bibinfo {author} {\bibfnamefont {N.}~\bibnamefont {Wyderka}},\ }\href {http://www.tp.nt.uni-siegen.de/+fhuber/ame.html} {\bibinfo {title} {Table of absolutely maximally entangled states}},\ \bibinfo {note} {accessed on 17 May., 2025}\BibitemShut {NoStop}%
\bibitem [{\citenamefont {Li}\ and\ \citenamefont {Wang}(2019{\natexlab{b}})}]{PhysRevA.99.042332}%
  \BibitemOpen
  \bibfield  {author} {\bibinfo {author} {\bibfnamefont {M.-S.}\ \bibnamefont {Li}}\ and\ \bibinfo {author} {\bibfnamefont {Y.-L.}\ \bibnamefont {Wang}},\ }\href {https://doi.org/10.1103/PhysRevA.99.042332} {\bibfield  {journal} {\bibinfo  {journal} {Phys. Rev. A}\ }\textbf {\bibinfo {volume} {99}},\ \bibinfo {pages} {042332} (\bibinfo {year} {2019}{\natexlab{b}})}\BibitemShut {NoStop}%
\bibitem [{\citenamefont {LaRacuente}\ and\ \citenamefont {Leditzky}(2024)}]{laracuente2024approximateunitarykdesignsshallow}%
  \BibitemOpen
  \bibfield  {author} {\bibinfo {author} {\bibfnamefont {N.}~\bibnamefont {LaRacuente}}\ and\ \bibinfo {author} {\bibfnamefont {F.}~\bibnamefont {Leditzky}},\ }\href@noop {} {\bibfield  {journal} {\bibinfo  {journal} {arXiv:2407.07876}\ } (\bibinfo {year} {2024})},\ \Eprint {https://arxiv.org/abs/2407.07876} {2407.07876 [quant-ph]} \BibitemShut {NoStop}%
\bibitem [{\citenamefont {Schuster}\ \emph {et~al.}(2024)\citenamefont {Schuster}, \citenamefont {Haferkamp},\ and\ \citenamefont {Huang}}]{schuster2024random}%
  \BibitemOpen
  \bibfield  {author} {\bibinfo {author} {\bibfnamefont {T.}~\bibnamefont {Schuster}}, \bibinfo {author} {\bibfnamefont {J.}~\bibnamefont {Haferkamp}},\ and\ \bibinfo {author} {\bibfnamefont {H.-Y.}\ \bibnamefont {Huang}},\ }\href@noop {} {\bibfield  {journal} {\bibinfo  {journal} {arXiv:2407.07754}\ } (\bibinfo {year} {2024})}\BibitemShut {NoStop}%
\bibitem [{\citenamefont {Hayden}\ \emph {et~al.}(2006)\citenamefont {Hayden}, \citenamefont {Leung},\ and\ \citenamefont {Winter}}]{Hayden_2006}%
  \BibitemOpen
  \bibfield  {author} {\bibinfo {author} {\bibfnamefont {P.}~\bibnamefont {Hayden}}, \bibinfo {author} {\bibfnamefont {D.~W.}\ \bibnamefont {Leung}},\ and\ \bibinfo {author} {\bibfnamefont {A.}~\bibnamefont {Winter}},\ }\href {https://doi.org/10.1007/s00220-006-1535-6} {\bibfield  {journal} {\bibinfo  {journal} {Commun. Math. Phys.}\ }\textbf {\bibinfo {volume} {265}},\ \bibinfo {pages} {95–117} (\bibinfo {year} {2006})}\BibitemShut {NoStop}%
\bibitem [{\citenamefont {Sen}(1996)}]{Sen_1996}%
  \BibitemOpen
  \bibfield  {author} {\bibinfo {author} {\bibfnamefont {S.}~\bibnamefont {Sen}},\ }\href {https://doi.org/10.1103/physrevlett.77.1} {\bibfield  {journal} {\bibinfo  {journal} {Phys. Rev. Lett.}\ }\textbf {\bibinfo {volume} {77}},\ \bibinfo {pages} {1–3} (\bibinfo {year} {1996})}\BibitemShut {NoStop}%
\bibitem [{\citenamefont {Bouland}\ \emph {et~al.}(2018)\citenamefont {Bouland}, \citenamefont {Fefferman}, \citenamefont {Nirkhe},\ and\ \citenamefont {Vazirani}}]{Bouland_2018}%
  \BibitemOpen
  \bibfield  {author} {\bibinfo {author} {\bibfnamefont {A.}~\bibnamefont {Bouland}}, \bibinfo {author} {\bibfnamefont {B.}~\bibnamefont {Fefferman}}, \bibinfo {author} {\bibfnamefont {C.}~\bibnamefont {Nirkhe}},\ and\ \bibinfo {author} {\bibfnamefont {U.}~\bibnamefont {Vazirani}},\ }\href {https://doi.org/10.1038/s41567-018-0318-2} {\bibfield  {journal} {\bibinfo  {journal} {Nature Physics}\ }\textbf {\bibinfo {volume} {15}},\ \bibinfo {pages} {159–163} (\bibinfo {year} {2018})}\BibitemShut {NoStop}%
\bibitem [{\citenamefont {Magesan}\ \emph {et~al.}(2011)\citenamefont {Magesan}, \citenamefont {Gambetta},\ and\ \citenamefont {Emerson}}]{PhysRevLett.106.180504}%
  \BibitemOpen
  \bibfield  {author} {\bibinfo {author} {\bibfnamefont {E.}~\bibnamefont {Magesan}}, \bibinfo {author} {\bibfnamefont {J.~M.}\ \bibnamefont {Gambetta}},\ and\ \bibinfo {author} {\bibfnamefont {J.}~\bibnamefont {Emerson}},\ }\href {https://doi.org/10.1103/PhysRevLett.106.180504} {\bibfield  {journal} {\bibinfo  {journal} {Phys. Rev. Lett.}\ }\textbf {\bibinfo {volume} {106}},\ \bibinfo {pages} {180504} (\bibinfo {year} {2011})}\BibitemShut {NoStop}%
\bibitem [{\citenamefont {Zhou}\ and\ \citenamefont {Hamma}(2022)}]{Zhou_2022}%
  \BibitemOpen
  \bibfield  {author} {\bibinfo {author} {\bibfnamefont {Y.}~\bibnamefont {Zhou}}\ and\ \bibinfo {author} {\bibfnamefont {A.}~\bibnamefont {Hamma}},\ }\href {https://doi.org/10.1103/PhysRevA.106.012410} {\bibfield  {journal} {\bibinfo  {journal} {Phys. Rev. A}\ }\textbf {\bibinfo {volume} {106}},\ \bibinfo {pages} {012410} (\bibinfo {year} {2022})}\BibitemShut {NoStop}%
\bibitem [{\citenamefont {Leone}\ \emph {et~al.}(2021)\citenamefont {Leone}, \citenamefont {Oliviero}, \citenamefont {Zhou},\ and\ \citenamefont {Hamma}}]{Leone_2021}%
  \BibitemOpen
  \bibfield  {author} {\bibinfo {author} {\bibfnamefont {L.}~\bibnamefont {Leone}}, \bibinfo {author} {\bibfnamefont {S.~F.~E.}\ \bibnamefont {Oliviero}}, \bibinfo {author} {\bibfnamefont {Y.}~\bibnamefont {Zhou}},\ and\ \bibinfo {author} {\bibfnamefont {A.}~\bibnamefont {Hamma}},\ }\href {https://doi.org/10.22331/q-2021-05-04-453} {\bibfield  {journal} {\bibinfo  {journal} {Quantum}\ }\textbf {\bibinfo {volume} {5}},\ \bibinfo {pages} {453} (\bibinfo {year} {2021})}\BibitemShut {NoStop}%
\bibitem [{\citenamefont {Webb}(2015)}]{Clifford3Design}%
  \BibitemOpen
  \bibfield  {author} {\bibinfo {author} {\bibfnamefont {Z.}~\bibnamefont {Webb}},\ }\href {https://doi.org/10.26421/QIC16.15-16-8} {\bibfield  {journal} {\bibinfo  {journal} {Quantum Information and Computation}\ }\textbf {\bibinfo {volume} {16}} (\bibinfo {year} {2015})}\BibitemShut {NoStop}%
\bibitem [{\citenamefont {Zhu}(2017)}]{Zhu_2017}%
  \BibitemOpen
  \bibfield  {author} {\bibinfo {author} {\bibfnamefont {H.}~\bibnamefont {Zhu}},\ }\bibfield  {journal} {\bibinfo  {journal} {Physical Review A}\ }\textbf {\bibinfo {volume} {96}},\ \href {https://doi.org/10.1103/physreva.96.062336} {10.1103/physreva.96.062336} (\bibinfo {year} {2017})\BibitemShut {NoStop}%
\bibitem [{\citenamefont {Zhu}\ \emph {et~al.}(2016)\citenamefont {Zhu}, \citenamefont {Kueng}, \citenamefont {Grassl},\ and\ \citenamefont {Gross}}]{zhu2016cliffordgroupfailsgracefully}%
  \BibitemOpen
  \bibfield  {author} {\bibinfo {author} {\bibfnamefont {H.}~\bibnamefont {Zhu}}, \bibinfo {author} {\bibfnamefont {R.}~\bibnamefont {Kueng}}, \bibinfo {author} {\bibfnamefont {M.}~\bibnamefont {Grassl}},\ and\ \bibinfo {author} {\bibfnamefont {D.}~\bibnamefont {Gross}},\ }\href@noop {} {\bibfield  {journal} {\bibinfo  {journal} {arXiv:2410.10116}\ } (\bibinfo {year} {2016})},\ \Eprint {https://arxiv.org/abs/1609.08172} {1609.08172 [quant-ph]} \BibitemShut {NoStop}%
\bibitem [{\citenamefont {Vedral}\ \emph {et~al.}(1997)\citenamefont {Vedral}, \citenamefont {Plenio}, \citenamefont {Jacobs},\ and\ \citenamefont {Knight}}]{Vedral_1997}%
  \BibitemOpen
  \bibfield  {author} {\bibinfo {author} {\bibfnamefont {V.}~\bibnamefont {Vedral}}, \bibinfo {author} {\bibfnamefont {M.~B.}\ \bibnamefont {Plenio}}, \bibinfo {author} {\bibfnamefont {K.}~\bibnamefont {Jacobs}},\ and\ \bibinfo {author} {\bibfnamefont {P.~L.}\ \bibnamefont {Knight}},\ }\href {https://doi.org/10.1103/physreva.56.4452} {\bibfield  {journal} {\bibinfo  {journal} {Physical Review A}\ }\textbf {\bibinfo {volume} {56}},\ \bibinfo {pages} {4452–4455} (\bibinfo {year} {1997})}\BibitemShut {NoStop}%
\bibitem [{\citenamefont {Vedral}(2002)}]{Vedral_2002}%
  \BibitemOpen
  \bibfield  {author} {\bibinfo {author} {\bibfnamefont {V.}~\bibnamefont {Vedral}},\ }\href {https://doi.org/10.1103/revmodphys.74.197} {\bibfield  {journal} {\bibinfo  {journal} {Reviews of Modern Physics}\ }\textbf {\bibinfo {volume} {74}},\ \bibinfo {pages} {197–234} (\bibinfo {year} {2002})}\BibitemShut {NoStop}%
\bibitem [{\citenamefont {Rains}(2000)}]{817508}%
  \BibitemOpen
  \bibfield  {author} {\bibinfo {author} {\bibfnamefont {E.}~\bibnamefont {Rains}},\ }\href {https://doi.org/10.1109/18.817508} {\bibfield  {journal} {\bibinfo  {journal} {IEEE Trans. Infor. Theory}\ }\textbf {\bibinfo {volume} {46}},\ \bibinfo {pages} {54} (\bibinfo {year} {2000})}\BibitemShut {NoStop}%
\bibitem [{\citenamefont {Page}(1993)}]{Page_1993}%
  \BibitemOpen
  \bibfield  {author} {\bibinfo {author} {\bibfnamefont {D.~N.}\ \bibnamefont {Page}},\ }\href {https://doi.org/10.1103/physrevlett.71.1291} {\bibfield  {journal} {\bibinfo  {journal} {Phys. Rev. Lett.}\ }\textbf {\bibinfo {volume} {71}},\ \bibinfo {pages} {1291–1294} (\bibinfo {year} {1993})}\BibitemShut {NoStop}%
\bibitem [{\citenamefont {Mele}(2024)}]{Mele_2024}%
  \BibitemOpen
  \bibfield  {author} {\bibinfo {author} {\bibfnamefont {A.~A.}\ \bibnamefont {Mele}},\ }\href {https://doi.org/10.22331/q-2024-05-08-1340} {\bibfield  {journal} {\bibinfo  {journal} {Quantum}\ }\textbf {\bibinfo {volume} {8}},\ \bibinfo {pages} {1340} (\bibinfo {year} {2024})}\BibitemShut {NoStop}%
\bibitem [{\citenamefont {Low}(2010)}]{low2010pseudorandomnesslearningquantumcomputation}%
  \BibitemOpen
  \bibfield  {author} {\bibinfo {author} {\bibfnamefont {R.~A.}\ \bibnamefont {Low}},\ }\href {https://arxiv.org/abs/1006.5227} {\bibinfo {title} {Pseudo-randomness and learning in quantum computation}} (\bibinfo {year} {2010}),\ \Eprint {https://arxiv.org/abs/1006.5227} {arXiv:1006.5227 [quant-ph]} \BibitemShut {NoStop}%
\bibitem [{\citenamefont {Low}(2009)}]{Low_2009}%
  \BibitemOpen
  \bibfield  {author} {\bibinfo {author} {\bibfnamefont {R.~A.}\ \bibnamefont {Low}},\ }\href {https://doi.org/10.1098/rspa.2009.0232} {\bibfield  {journal} {\bibinfo  {journal} {P. Roy. Soc. A-Math. Phys.}\ }\textbf {\bibinfo {volume} {465}},\ \bibinfo {pages} {3289–3308} (\bibinfo {year} {2009})}\BibitemShut {NoStop}%
\bibitem [{\citenamefont {Huber}\ and\ \citenamefont {Grassl}(2020)}]{Huber_2020}%
  \BibitemOpen
  \bibfield  {author} {\bibinfo {author} {\bibfnamefont {F.}~\bibnamefont {Huber}}\ and\ \bibinfo {author} {\bibfnamefont {M.}~\bibnamefont {Grassl}},\ }\href {https://doi.org/10.22331/q-2020-06-18-284} {\bibfield  {journal} {\bibinfo  {journal} {Quantum}\ }\textbf {\bibinfo {volume} {4}},\ \bibinfo {pages} {284} (\bibinfo {year} {2020})}\BibitemShut {NoStop}%
\bibitem [{\citenamefont {Schumacher}\ and\ \citenamefont {Westmoreland}(2002)}]{schumacher2001approximatequantumerrorcorrection}%
  \BibitemOpen
  \bibfield  {author} {\bibinfo {author} {\bibfnamefont {B.}~\bibnamefont {Schumacher}}\ and\ \bibinfo {author} {\bibfnamefont {M.~D.}\ \bibnamefont {Westmoreland}},\ }\href {https://link.springer.com/article/10.1023/A:1019653202562} {\bibfield  {journal} {\bibinfo  {journal} {Quantum Infor. Process.}\ }\textbf {\bibinfo {volume} {1}},\ \bibinfo {pages} {5} (\bibinfo {year} {2002})}\BibitemShut {NoStop}%
\bibitem [{\citenamefont {Leung}\ \emph {et~al.}(1997)\citenamefont {Leung}, \citenamefont {Nielsen}, \citenamefont {Chuang},\ and\ \citenamefont {Yamamoto}}]{Leung_1997}%
  \BibitemOpen
  \bibfield  {author} {\bibinfo {author} {\bibfnamefont {D.~W.}\ \bibnamefont {Leung}}, \bibinfo {author} {\bibfnamefont {M.~A.}\ \bibnamefont {Nielsen}}, \bibinfo {author} {\bibfnamefont {I.~L.}\ \bibnamefont {Chuang}},\ and\ \bibinfo {author} {\bibfnamefont {Y.}~\bibnamefont {Yamamoto}},\ }\href {https://doi.org/10.1103/physreva.56.2567} {\bibfield  {journal} {\bibinfo  {journal} {Phys. Rev. A}\ }\textbf {\bibinfo {volume} {56}},\ \bibinfo {pages} {2567–2573} (\bibinfo {year} {1997})}\BibitemShut {NoStop}%
\bibitem [{\citenamefont {Liu}\ and\ \citenamefont {Zhou}(2023)}]{Liu_2023}%
  \BibitemOpen
  \bibfield  {author} {\bibinfo {author} {\bibfnamefont {Z.-W.}\ \bibnamefont {Liu}}\ and\ \bibinfo {author} {\bibfnamefont {S.}~\bibnamefont {Zhou}},\ }\href {http://dx.doi.org/10.1038/s41534-023-00788-4} {\bibfield  {journal} {\bibinfo  {journal} {npj Quantum Infor.}\ }\textbf {\bibinfo {volume} {9}},\ \bibinfo {pages} {119} (\bibinfo {year} {2023})}\BibitemShut {NoStop}%
\bibitem [{\citenamefont {Yi}\ \emph {et~al.}(2024)\citenamefont {Yi}, \citenamefont {Ye}, \citenamefont {Gottesman},\ and\ \citenamefont {Liu}}]{Yi_2024}%
  \BibitemOpen
  \bibfield  {author} {\bibinfo {author} {\bibfnamefont {J.}~\bibnamefont {Yi}}, \bibinfo {author} {\bibfnamefont {W.}~\bibnamefont {Ye}}, \bibinfo {author} {\bibfnamefont {D.}~\bibnamefont {Gottesman}},\ and\ \bibinfo {author} {\bibfnamefont {Z.-W.}\ \bibnamefont {Liu}},\ }\href {https://doi.org/10.1038/s41567-024-02621-x} {\bibfield  {journal} {\bibinfo  {journal} {Nat. Phys.}\ }\textbf {\bibinfo {volume} {20}},\ \bibinfo {pages} {1798–1803} (\bibinfo {year} {2024})}\BibitemShut {NoStop}%
\bibitem [{\citenamefont {Liu}\ \emph {et~al.}(2025)\citenamefont {Liu}, \citenamefont {Du}, \citenamefont {Liu},\ and\ \citenamefont {Ma}}]{liu2025approximatequantumerrorcorrection}%
  \BibitemOpen
  \bibfield  {author} {\bibinfo {author} {\bibfnamefont {G.}~\bibnamefont {Liu}}, \bibinfo {author} {\bibfnamefont {Z.}~\bibnamefont {Du}}, \bibinfo {author} {\bibfnamefont {Z.-W.}\ \bibnamefont {Liu}},\ and\ \bibinfo {author} {\bibfnamefont {X.}~\bibnamefont {Ma}},\ }\href@noop {} {\bibfield  {journal} {\bibinfo  {journal} {arXiv:2503.17759}\ } (\bibinfo {year} {2025})},\ \Eprint {https://arxiv.org/abs/2503.17759} {2503.17759 [quant-ph]} \BibitemShut {NoStop}%
\bibitem [{\citenamefont {Shor}\ and\ \citenamefont {Laflamme}(1997)}]{shor1996quantummacwilliamsidentities}%
  \BibitemOpen
  \bibfield  {author} {\bibinfo {author} {\bibfnamefont {P.}~\bibnamefont {Shor}}\ and\ \bibinfo {author} {\bibfnamefont {R.}~\bibnamefont {Laflamme}},\ }\href {https://doi.org/10.1103/PhysRevLett.78.1600} {\bibfield  {journal} {\bibinfo  {journal} {Phys. Rev. Lett.}\ }\textbf {\bibinfo {volume} {78}},\ \bibinfo {pages} {1600} (\bibinfo {year} {1997})}\BibitemShut {NoStop}%
\bibitem [{\citenamefont {Rains}(1998)}]{rains1996quantumweightenumerators}%
  \BibitemOpen
  \bibfield  {author} {\bibinfo {author} {\bibfnamefont {E.}~\bibnamefont {Rains}},\ }\href {https://doi.org/10.1109/18.681316} {\bibfield  {journal} {\bibinfo  {journal} {IEEE Trans. Inf. Theory}\ }\textbf {\bibinfo {volume} {44}},\ \bibinfo {pages} {1388} (\bibinfo {year} {1998})}\BibitemShut {NoStop}%
\bibitem [{\citenamefont {Miller}\ \emph {et~al.}(2024)\citenamefont {Miller}, \citenamefont {Levi}, \citenamefont {Postler}, \citenamefont {Steiner}, \citenamefont {Bittel}, \citenamefont {White}, \citenamefont {Tang}, \citenamefont {Kuehnke}, \citenamefont {Mele}, \citenamefont {Khatri}, \citenamefont {Leone}, \citenamefont {Carrasco}, \citenamefont {Marciniak}, \citenamefont {Pogorelov}, \citenamefont {Guevara-Bertsch}, \citenamefont {Freund}, \citenamefont {Blatt}, \citenamefont {Schindler}, \citenamefont {Monz}, \citenamefont {Ringbauer},\ and\ \citenamefont {Eisert}}]{miller2024experimentalmeasurementphysicalinterpretation}%
  \BibitemOpen
  \bibfield  {author} {\bibinfo {author} {\bibfnamefont {D.}~\bibnamefont {Miller}}, \bibinfo {author} {\bibfnamefont {K.}~\bibnamefont {Levi}}, \bibinfo {author} {\bibfnamefont {L.}~\bibnamefont {Postler}}, \bibinfo {author} {\bibfnamefont {A.}~\bibnamefont {Steiner}}, \bibinfo {author} {\bibfnamefont {L.}~\bibnamefont {Bittel}}, \bibinfo {author} {\bibfnamefont {G.~A.~L.}\ \bibnamefont {White}}, \bibinfo {author} {\bibfnamefont {Y.}~\bibnamefont {Tang}}, \bibinfo {author} {\bibfnamefont {E.~J.}\ \bibnamefont {Kuehnke}}, \bibinfo {author} {\bibfnamefont {A.~A.}\ \bibnamefont {Mele}}, \bibinfo {author} {\bibfnamefont {S.}~\bibnamefont {Khatri}}, \bibinfo {author} {\bibfnamefont {L.}~\bibnamefont {Leone}}, \bibinfo {author} {\bibfnamefont {J.}~\bibnamefont {Carrasco}}, \bibinfo {author} {\bibfnamefont {C.~D.}\ \bibnamefont {Marciniak}}, \bibinfo {author} {\bibfnamefont {I.}~\bibnamefont {Pogorelov}}, \bibinfo {author} {\bibfnamefont {M.}~\bibnamefont {Guevara-Bertsch}}, \bibinfo {author} {\bibfnamefont
  {R.}~\bibnamefont {Freund}}, \bibinfo {author} {\bibfnamefont {R.}~\bibnamefont {Blatt}}, \bibinfo {author} {\bibfnamefont {P.}~\bibnamefont {Schindler}}, \bibinfo {author} {\bibfnamefont {T.}~\bibnamefont {Monz}}, \bibinfo {author} {\bibfnamefont {M.}~\bibnamefont {Ringbauer}},\ and\ \bibinfo {author} {\bibfnamefont {J.}~\bibnamefont {Eisert}},\ }\href {https:/ /arxiv.org/abs/2408.16914} {\bibfield  {journal} {\bibinfo  {journal} {arXiv:2408.16914}\ } (\bibinfo {year} {2024})}\BibitemShut {NoStop}%
\bibitem [{\citenamefont {Shi}\ \emph {et~al.}(2024)\citenamefont {Shi}, \citenamefont {Guo}, \citenamefont {Zhang},\ and\ \citenamefont {Zhao}}]{shi2024exploringquantumweightenumerators}%
  \BibitemOpen
  \bibfield  {author} {\bibinfo {author} {\bibfnamefont {F.}~\bibnamefont {Shi}}, \bibinfo {author} {\bibfnamefont {K.}~\bibnamefont {Guo}}, \bibinfo {author} {\bibfnamefont {X.}~\bibnamefont {Zhang}},\ and\ \bibinfo {author} {\bibfnamefont {Q.}~\bibnamefont {Zhao}},\ }\href {https://arxiv.org/abs/2406.18280} {\bibfield  {journal} {\bibinfo  {journal} {arXiv:2406.18280}\ } (\bibinfo {year} {2024})}\BibitemShut {NoStop}%
\bibitem [{\citenamefont {Hayden}\ \emph {et~al.}(2004)\citenamefont {Hayden}, \citenamefont {Leung}, \citenamefont {Shor},\ and\ \citenamefont {Winter}}]{Hayden_2004}%
  \BibitemOpen
  \bibfield  {author} {\bibinfo {author} {\bibfnamefont {P.}~\bibnamefont {Hayden}}, \bibinfo {author} {\bibfnamefont {D.}~\bibnamefont {Leung}}, \bibinfo {author} {\bibfnamefont {P.~W.}\ \bibnamefont {Shor}},\ and\ \bibinfo {author} {\bibfnamefont {A.}~\bibnamefont {Winter}},\ }\href {https://doi.org/10.1007/s00220-004-1087-6} {\bibfield  {journal} {\bibinfo  {journal} {Communications in Mathematical Physics}\ }\textbf {\bibinfo {volume} {250}},\ \bibinfo {pages} {371–391} (\bibinfo {year} {2004})}\BibitemShut {NoStop}%
\bibitem [{\citenamefont {Modi}\ \emph {et~al.}(2018)\citenamefont {Modi}, \citenamefont {Pati}, \citenamefont {Sen(De)},\ and\ \citenamefont {Sen}}]{PhysRevLett.120.230501}%
  \BibitemOpen
  \bibfield  {author} {\bibinfo {author} {\bibfnamefont {K.}~\bibnamefont {Modi}}, \bibinfo {author} {\bibfnamefont {A.~K.}\ \bibnamefont {Pati}}, \bibinfo {author} {\bibfnamefont {A.}~\bibnamefont {Sen(De)}},\ and\ \bibinfo {author} {\bibfnamefont {U.}~\bibnamefont {Sen}},\ }\href {https://doi.org/10.1103/PhysRevLett.120.230501} {\bibfield  {journal} {\bibinfo  {journal} {Phys. Rev. Lett.}\ }\textbf {\bibinfo {volume} {120}},\ \bibinfo {pages} {230501} (\bibinfo {year} {2018})}\BibitemShut {NoStop}%
\bibitem [{\citenamefont {Pérez-García}\ \emph {et~al.}(2006)\citenamefont {Pérez-García}, \citenamefont {Wolf}, \citenamefont {Petz},\ and\ \citenamefont {Ruskai}}]{10.1063/1.2218675}%
  \BibitemOpen
  \bibfield  {author} {\bibinfo {author} {\bibfnamefont {D.}~\bibnamefont {Pérez-García}}, \bibinfo {author} {\bibfnamefont {M.~M.}\ \bibnamefont {Wolf}}, \bibinfo {author} {\bibfnamefont {D.}~\bibnamefont {Petz}},\ and\ \bibinfo {author} {\bibfnamefont {M.~B.}\ \bibnamefont {Ruskai}},\ }\href {https://doi.org/10.1063/1.2218675} {\bibfield  {journal} {\bibinfo  {journal} {J. Math. Phys}\ }\textbf {\bibinfo {volume} {47}},\ \bibinfo {pages} {083506} (\bibinfo {year} {2006})}\BibitemShut {NoStop}%
\bibitem [{\citenamefont {Harrow}\ \emph {et~al.}(2004)\citenamefont {Harrow}, \citenamefont {Hayden},\ and\ \citenamefont {Leung}}]{Harrow_2004}%
  \BibitemOpen
  \bibfield  {author} {\bibinfo {author} {\bibfnamefont {A.}~\bibnamefont {Harrow}}, \bibinfo {author} {\bibfnamefont {P.}~\bibnamefont {Hayden}},\ and\ \bibinfo {author} {\bibfnamefont {D.}~\bibnamefont {Leung}},\ }\href {http://dx.doi.org/10.1103/PhysRevLett.92.187901} {\bibfield  {journal} {\bibinfo  {journal} {Phys. Rev. Lett.}\ }\textbf {\bibinfo {volume} {92}} (\bibinfo {year} {2004})}\BibitemShut {NoStop}%
\bibitem [{\citenamefont {Ma}\ and\ \citenamefont {Huang}(2025)}]{ma2025constructrandomunitaries}%
  \BibitemOpen
  \bibfield  {author} {\bibinfo {author} {\bibfnamefont {F.}~\bibnamefont {Ma}}\ and\ \bibinfo {author} {\bibfnamefont {H.-Y.}\ \bibnamefont {Huang}},\ }\href@noop {} {\bibfield  {journal} {\bibinfo  {journal} {arXiv:2410.10116}\ } (\bibinfo {year} {2025})},\ \Eprint {https://arxiv.org/abs/2410.10116} {2410.10116 [quant-ph]} \BibitemShut {NoStop}%
\end{thebibliography}%
\end{document}